\documentclass[lettersize,journal]{IEEEtran}
\usepackage{subcaption}
\usepackage{xspace}
\usepackage{rotating}
\usepackage{amsmath,amssymb,amsfonts}
\usepackage{mathalfa}
\usepackage{algorithmic}
\usepackage{graphicx}
\usepackage{pifont}
\usepackage{textcomp}
\usepackage[ruled,linesnumbered]{algorithm2e}
\usepackage{xcolor}
\usepackage{lipsum}
\usepackage{booktabs}
\usepackage{adjustbox}
\usepackage{caption}
\usepackage[most]{tcolorbox}
\usepackage{colortbl}
\usepackage{xcolor}
\usepackage{adjustbox}
\usepackage{xcolor}
\usepackage{tikz}
\usetikzlibrary{shapes.geometric}
\usepackage{bbding}
\usepackage{MnSymbol}
\usepackage{wasysym}
\usepackage[colorlinks=true, linkcolor=blue]{hyperref}
\usepackage{multirow}
\usepackage[table]{xcolor} 
\usepackage{colortbl}
\usepackage{booktabs}
\usepackage{float}
\usepackage{dblfloatfix}
\usepackage{amsthm}

\definecolor{outer20}{RGB}{241,106,126}
\definecolor{outer24}{RGB}{188,152,48}
\definecolor{outer28}{RGB}{96,176,64}
\definecolor{outer32}{RGB}{64,168,150}
\definecolor{outer36}{RGB}{74,144,226}
\definecolor{outer40}{RGB}{214,102,214}

\definecolor{inner22}{RGB}{241,106,126}
\definecolor{inner26}{RGB}{188,152,48}
\definecolor{inner30}{RGB}{96,176,64}
\definecolor{inner34}{RGB}{64,168,150}
\definecolor{inner38}{RGB}{74,144,226}
\usepackage{amssymb}

\usepackage{amsmath}
\DeclareMathOperator*{\argmin}{arg\,min}

\renewcommand{\tabcolsep}{2pt}

\newcommand{\localupdate}{{\ensuremath{\sf{\mathsf ClientUpdate}}}\xspace}
\newcommand{\aggregate}{{\ensuremath{\sf{\mathsf ServerAgg}}}\xspace}
\newcommand{\keygen}{{\ensuremath{\sf{\mathsf CKKS.KeyGen}}}\xspace}
\newcommand{\encode}{{\ensuremath{\sf{\mathsf CKKS.Encode}}}\xspace}
\newcommand{\encrypt}{{\ensuremath{\sf{\mathsf CKKS.Encrypt}}}\xspace}
\newcommand{\decode}{{\ensuremath{\sf{\mathsf CKKS.Decode}}}\xspace}
\newcommand{\decrypt}{{\ensuremath{\sf{\mathsf CKKS.Decrypt}}}\xspace}
\newcommand{\add}{{\ensuremath{\sf{\mathsf CKKS.Add}}}\xspace}

\newcommand{\mult}{{\ensuremath{\sf{\mathsf CKKS.Mult}}}\xspace}
\newcommand{\multconst}{{\ensuremath{\sf{\mathsf CKKS.MultConst}}}\xspace}
\newcommand{\acro}{pFedCKKS\xspace}
\newcommand\ignore[1]{}

\newtheorem{lemma}{Lemma}
\newcommand{\finetune}{{\ensuremath{\sf{\mathsf FedFinetune}}}\xspace}
\newcommand{\fedper}{{\ensuremath{\sf{\mathsf FedPer}}}\xspace}
\newcommand{\ditto}{{\ensuremath{\sf{\mathsf Ditto}}}\xspace}
\newcommand{\blackpentagon}{\mbox{%
  \begin{tikzpicture}[baseline=-0.6ex]
    \fill (90:3.6pt) -- (162:3.6pt) -- (234:3.6pt) -- (306:3.6pt) -- (18:3.6pt) -- cycle;
  \end{tikzpicture}%
}}

\newtcbtheorem{Takeaway}{Takeaway}{
  enhanced,
  breakable, 
  drop shadow={black!50!white},
  coltitle=black,
  fonttitle=\bfseries,
  top=0.2in,
  attach boxed title to top left=
  {xshift=1.5em,yshift=-\tcboxedtitleheight/2},
  boxed title style={size=small,colback=pink}
}{takeaway}

\begin{document}


\title{Exploring CKKS Parameter Trade-offs for Privacy-Preserving Personalized Federated Learning}


\author{
Kamolchanok~Saengtong,
Phanwadee~Sinthong,
Norrathep~Rattanavipanon%
\thanks{Kamolchanok Saengtong and Norrathep Rattanavipanon are with the College of Computing, Prince of Songkla University, Phuket, Thailand.}%
\thanks{Phanwadee Sinthong is with the School of Informatics, Walailak University, Nakhon Si Thammarat, Thailand.}%
\thanks{Corresponding authors: Phanwadee Sinthong (phanwadee.si@wu.ac.th) and Norrathep Rattanavipanon (norrathep.r@psu.ac.th).}
}

\maketitle


\begin{abstract}
Privacy-preserving Personalized Federated Learning (PFL) enables clients to collaboratively train personalized models without exposing raw data, but exchanged model updates remain vulnerable to inference attacks from honest-but-curious servers. 
Homomorphic Encryption (HE) addresses this by allowing server-side aggregation directly on encrypted updates, with the CKKS scheme being particularly suitable due to its native support for approximate floating-point arithmetic. However, no prior work has examined how to configure CKKS for PFL deployments, leaving practitioners without principled guidance on parameter selection that directly affects privacy, precision, and computational cost.

This paper presents \acro, a generic framework integrating CKKS into PFL, and provides the first systematic parameter selection guide for practitioners. We derive the full CKKS parameter constraints under 128-bit security for the PFL setting, showing the selection problem reduces to choosing just two values: the inner and outer ciphertext prime. 
Implemented using the Flower framework and TenSEAL library, pFedCKKS is evaluated on the FEMNIST, CelebA and Sentiment140 datasets with \finetune, \ditto and \fedper which represents PFL algorithms. 
Experimental results reveal an empirical trade-off between precision and computational/communication costs.
This allows us to draw a concrete guideline for selecting proper CKKS parameters that balance efficiency and accuracy in real-world deployments of \acro. 
\end{abstract}



\section{Introduction}
The growing demand for data-driven Artificial Intelligence technologies has led to large-scale personal data collection, which in turn stimulates serious privacy concerns. 
In response, many countries have enacted privacy regulations (e.g., GDPR) that limit the sharing and use of personal information. 
To reconcile the conflict between data utility and privacy, Federated Learning (FL)~\cite{mcmahan2017communication} has emerged as a promising solution. 
As shown in Figure~\ref{fig:setting}, instead of transmitting raw data to a central server, FL allows each client to {\color{red}\ding{192}} train a local machine learning (ML) model and {\color{red}\ding{193}} share only model updates (e.g., gradients or parameters) with the server. 
The server then aggregates these local updates into a global model {\color{red}\ding{194}} and redistribute it to clients for the next training round {\color{red}\ding{195}}. 
This process is repeated until the global model converges or a predefined number of rounds is reached, after which it produces the final global model that can be shared with all participating clients. 
By keeping data on-device, FL reduces privacy risks while allowing organizations to work with sensitive data to comply with regulatory requirements.

A key practical limitation of standard FL is that all clients ultimately obtains the same ML model, i.e., the one in a red box of Figure~\ref{fig:setting}. 
This can lead to suboptimal performance in real-world scenarios, where client data is typically heterogeneous and non-independent and identically distributed (non-IID)~\cite{zhao2018federated}.  
In such settings, a single ML model from FL cannot effectively capture all distributions across clients. 
To address this, Personalized Federated Learning (PFL)~\cite{tan2022towards} has been proposed. 
Unlike traditional FL, PFL enables each client to develop a personalized model (i.e., the ones in green dotted boxes of Figure~\ref{fig:setting}) tailored to its local data.
As a result, PFL has been shown to improve performance in practical non-IID settings~\cite{arivazhagan2019federated,li2021ditto,collins2022fedavg,tan2022towards,lu2024federated,fallah2020personalized,hanzely2020federated}. 

\begin{figure}[h]
    \centering
    \includegraphics[width=\columnwidth]{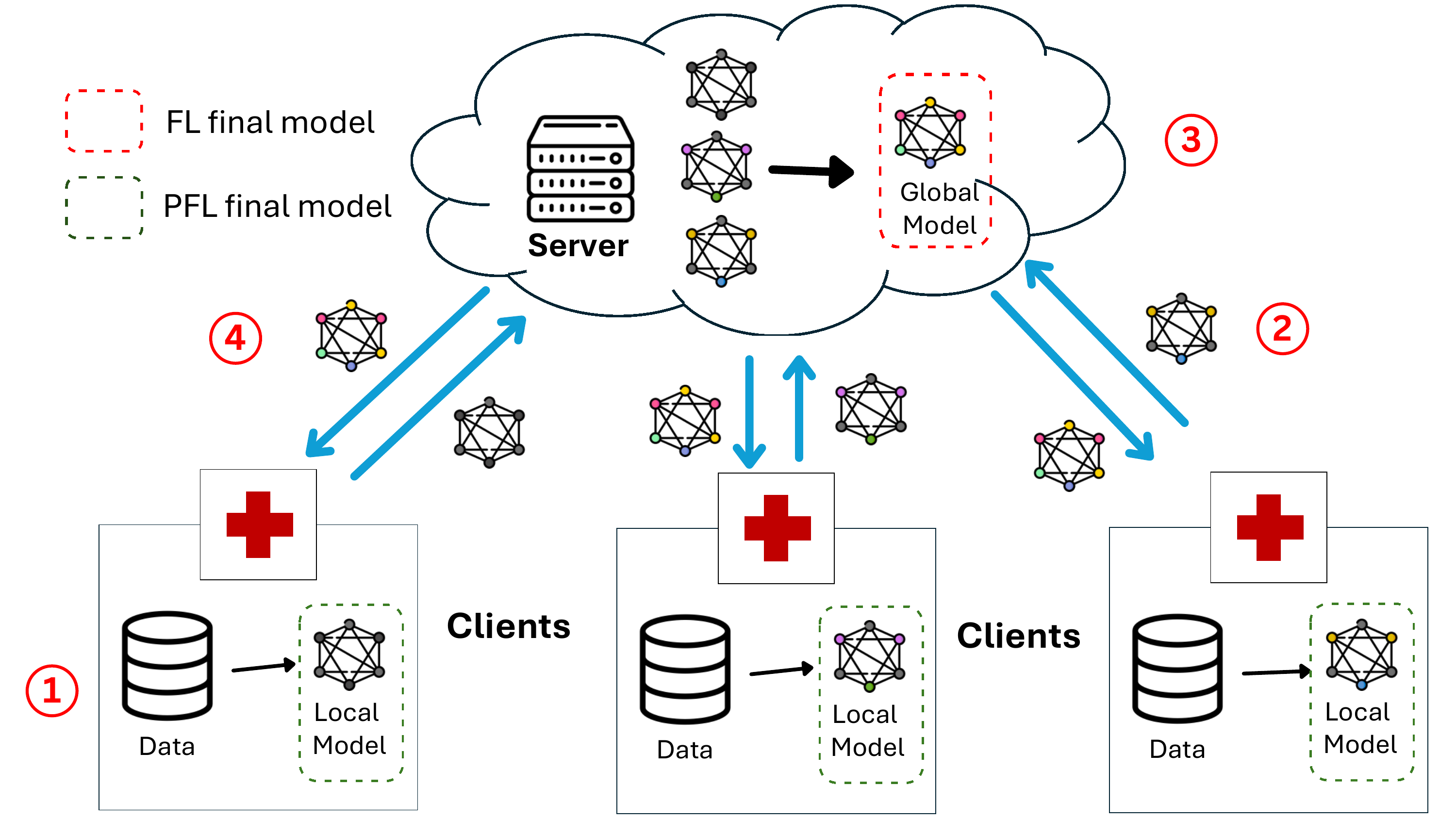}
    \caption{Federated Learning (FL) vs Personalized Federated Learning (PFL)}
    \label{fig:setting}
\end{figure}

While FL/PFL reduce privacy risks compared to the centralized approach, they are still not immune to privacy leakage.
Prior studies have demonstrated that local updates still retain sensitive information of the training data, making them vulnerable to inference attacks. 
These attacks, for example, can enable the central server to reconstruct private training data~\cite{zhu2019deep,wang2019beyond}, infer sensitive properties of the training dataset~\cite{melis2019exploiting,luo2021feature}, or determine whether a specific record (i.e., membership) was used during training~\cite{nasr2019comprehensive,shokri2017membership}.



To mitigate these risks, various techniques have been proposed, among which is homomorphic encryption (HE).
HE enables certain computations to be performed directly on encrypted gradients or model parameters without requiring decryption.
This allows the central server to aggregate encrypted local updates without ever seeing the underlying plaintexts.
Hence, the global model can be computed at the ciphertext-level on the central server and is only decryptable on the client's side, preventing the central server from performing the privacy attacks on local updates.
Among existing HE schemes, the CKKS (Cheon-Kim-Kim-Song) scheme~\cite{cheon2017homomorphic,cheon2018full} has become the preferred choice of HE for machine learning tasks since it efficiently supports approximate arithmetic over real numbers, which are the standard data type for gradients or model parameters.
This has led to several studies exploring integration of CKKS into FL to enhance its privacy~\cite{pan2024fedshe,yao2023efficient,qiu2022privacy}.
However, to the best of our knowledge, no prior work has investigated the use of CKKS in the context of PFL.
Incorporating CKKS into PFL poses non-trivial practical challenges related to parameter selection, where
these parameters can significantly impact various performance aspects such as encryption/decryption time, ciphertext size, and the precision of decrypted results.
This complexity leaves open the question of how best to configure and optimize CKKS for PFL settings, where heterogeneous computing resources and network conditions have a direct effect on model training and performance.

To bridge this gap, this paper first presents the integration of the CKKS scheme into PFL and then conducts an extensive evaluation by exploring key factors to assess the performance impact introduced by CKKS. The main contributions of this work are as follows:

\begin{itemize}
    \item \textbf{New Framework.} We design a generic framework, called \acro, which integrates PFL with the CKKS scheme. 
    To demonstrate its compatibility with major PFL algorithms, we apply our integration to popular PFL strategies: FedAVG+Finetuned~\cite{collins2022fedavg}, FedPer~\cite{arivazhagan2019federated} and Ditto~\cite{li2021ditto}. 
    These three collectively represent major PFL categories according to a recent PFL taxonomy~\cite{fan2024lightweight}.

    \item \textbf{Analysis of CKKS Parameters.} To quantify the impact of this integration, we first analyze the full set of CKKS parameters used in \acro under the standard security level. Our analysis reveals that the CKKS ciphertext moduli are the most influential, as it strongly controls \acro soundness. Our analysis also reveals constraints on how the modulus values should be selected to achieve the standard security level.

    \item \textbf{Empirical Evaluation.} Given the constraints, we evaluate \acro across a wide range of ciphertext modulus values to identify the combination that yields optimal performance; our experiments focus on three metrics: computation, communication, and precision, using three different datasets (two vision and one text classification tasks) 
    The results indicate that ciphertext moduli of $(28,26,28)$ (i.e., 28-bit outer primes and 26-bit inner prime) provide a practical operating point across all datasets, yielding the best balance between precision and performance.
    
\end{itemize}

\textbf{Organization:}  
The rest of this paper is organized into the following sections:
\hyperref[sec:preliminaries]{Section~\ref{sec:preliminaries}} introduces background on PFL algorithms used in this work (\finetune, \ditto, \fedper) and the CKKS encryption scheme. 
\hyperref[sec:framework]{Section~\ref{sec:framework}} describes the \acro system model, adversary assumptions, workflow, and analysis. 
\hyperref[sec:experiments]{Section~\ref{sec:experiments}} details the experimental setup and evaluates performance. 
\hyperref[sec:related]{Section~\ref{sec:related}} reviews prior work on federated learning and personalized FL. 
Finally, the paper concludes in \hyperref[sec:conclusion]{Section~\ref{sec:conclusion}}.

\section{Preliminaries}
\label{sec:preliminaries}


\subsection{Personalized Federated Learning}\label{subsec:pfl}

In PFL, we consider a distributed system consisting of one central server and $n$ clients.
Each client $C_i$ holds a private dataset $D_i$. 
PFL operates in multiple rounds where each round runs two algorithms in sequence:

\begin{itemize}
    \item $\localupdate(\Theta, \theta_i, D_i) \rightarrow \hat{\Theta}_i, \hat{\theta_i}$: Executed on client $C_i$, this algorithm takes as input global parameters $\Theta$ (shared for all clients) and personalized parameters $\theta_i$ (specific to $C_i$) and then trains on $D_i$, producing updated parameters $\hat{\Theta_i}$ and $\hat{\theta_i}$. $\hat{\Theta_i}$ is sent to the server while $C_i$ keeps $\hat{\theta_i}$ locally used as $\theta_i$ in the next round.

    \item $\aggregate(\hat{\Theta}_1, ..., \hat{\Theta}_n) \rightarrow \Theta$: Executed on the central server, this algorithm aggregates locally updated shareable parameters from all clients ($\hat{\Theta}_1, ..., \hat{\Theta}_n$) into a new global $\Theta$, which is then sent to all clients for the next round's \localupdate.
\end{itemize}

The objective of PFL is to find personalized parameters $\hat{\theta}_1, ..., \hat{\theta}_n$ that minimizes the following global objective:

\begin{center}
    $\hat{\theta}_1, ..., \hat{\theta}_n = \argmin G(L_1, ..., L_n)$
\end{center}

where $L_i = L(\theta_i,D_i;\Theta)$ is the local loss on client $C_i$, $L$ denotes the loss function, and $G$ is commonly computed as a weighted average of local losses w.r.t. the dataset sizes:

\begin{center}
    $G(L_1,...,L_n) = \sum_{i=1}^n L_i \cdot \frac{|D_i|}{\sum_{j=1}^n|D_j|}$
\end{center}

Let $\phi$ denote as a null/unused variable.
Next, we describe how \localupdate and \aggregate can be instantiated in PFL algorithms experimented in this work, where the choice of these algorithms is motivated by the recent PFL taxonomy~\cite{fan2024lightweight}.

\textbf{FedAVG+Finetuned (\finetune)~\cite{collins2022fedavg}} extends FedAVG~\cite{mcmahan2017communication} to support non-IID settings by adopting personalized parameters from fine-tuned models rather than directly using the global model from FedAVG. 
Specifically, \finetune implements each round's algorithm as:

\begin{itemize}
    \item $\finetune.\localupdate(\Theta, \phi, D_i)$: Client $C_i$ receives the entire model parameters $\Theta$ from the server (randomly initialized by the server in the first round) and performs stochastic gradient descent (SGD) using $\Theta$ on $D_i$ for a fixed number of steps. It then produces updated parameters $\hat{\Theta}_i$. If this is a final round, it sets $\hat{\theta}_i$ as $\hat{\Theta}_i$ to serve as the client's personalized model; otherwise $\hat{\theta}_i = \phi$.

    \item $\finetune.\aggregate(\hat{\Theta}_1, ..., \hat{\Theta}_n)$: The server aggregates all $\hat{\Theta}$-s using a weighted average:
    \begin{equation}\label{eq:weighted-avg}
    \Theta \leftarrow \sum_{i=1}^n \Theta_i \cdot \frac{|D_i|}{\sum_{j=1}^n|D_j|}
    \end{equation}
\end{itemize}

\noindent\textbf{FedPer~\cite{arivazhagan2019federated}} splits training of a deep learning model into two disjoint parts: the base layers (i.e., early layers) are shared with the central server, while the personalized layers are kept private on each client.

\begin{itemize}
    \item $\mathsf{FedPer}.\localupdate(\Theta, \theta_i, D_i)$: Client $C_i$ receives the global base layers $\Theta$ from the server (random weights in the first rounds) and personalized parameters $\theta_i$ (i.e., $\hat{\theta}_i$ from the previous round). It then performs SGD on the combined model ($\Theta, \theta_i$) using $D_i$ for a fixed number steps. This produces updated parameters for the shared base layers $\hat{\Theta}_i$ and the personalized layers $\hat{\theta}_i$.

    \item $\mathsf{FedPer}.\aggregate(\hat{\Theta}_1, ..., \hat{\Theta}_n)$: A common approach in FedPer is to aggregate all shared base layers via a weighted average, following Equation~\ref{eq:weighted-avg}.
\end{itemize}

\noindent\textbf{Ditto~\cite{li2021ditto}.} Unlike previous strategies, \ditto completely separates the entire personalized model from the global one and introduces an additional proximal term when training the personalized parameters.

\begin{itemize}
    \item $\mathsf{Ditto}.\localupdate(\Theta, \phi, D_i)$: Client $C_i$ receives the global parameters $\Theta$ from the server (random weights in the first rounds) and performs SGD using $\Theta$ on $D_i$ for a fixed number steps to obtain $\hat{\Theta}$. Then, it derives $\theta_i$ from $\Theta$ by solving the following optimization problem:

    \begin{equation}
        \theta_i \leftarrow \argmin_{\theta}(L(\theta,D_i)+\frac{\lambda}{2}||\theta-\Theta||^2)
    \end{equation}

    where $\lambda$ is a hyperparameter in which higher $\lambda$ leads to $\theta_i$ being closer to the global model.
    
    \item $\mathsf{Ditto}.\aggregate(\hat{\Theta}_1, ..., \hat{\Theta}_n)$: Similar to \finetune and \fedper, a weighted average is commonly a preferred choice for aggregation on the server, i.e., this algorithm follows Equation~\ref{eq:weighted-avg}.
\end{itemize}


\subsection{Cheon-Kim-Kim-Song (CKKS) Scheme}

The main ``claim to fame'' of the CKKS scheme is its ability to perform HE for approximate arithmetic that natively supports \emph{complex numbers}. This makes this scheme well-suited for floating-point data such as model parameters and gradients.
In addition, CKKS features Single Instruction Multiple Data (SIMD) parallelism; this allows a vector of floating-point data to be packed into a single ciphertext and processed simultaneously.

Security of this scheme relies on the computational hardness assumption of the Ring-based Learning With Errors (RLWE) problem~\cite{lyubashevsky2010ideal}, which encrypts a message under noisy inner products.
In particular, CKKS is defined over the polynomial ring $\mathcal{R}$ with integer coefficients, i.e., $\mathcal{R} = \mathbb{Z}[X]/(X^N+1)$ where $N$ is the degree of the polynomial modulus (typically a power of two).
The scheme supports a pre-defined number $L$ of multiplicative depths (or levels), with each level $0 < l \leq  L$ associated with a ciphertext prime modulus $q_l$.
Besides $q_l$-s, $q_0$ is used for encryption and $q_{L+1}$ represents a special last prime to ensure sufficient precision for decryption.
Intuitively, $L$ specifies how many times homomorphic multiplication can be performed in the current CKKS setup.
We define the total modulus at level $l$ as $Q_{l} = q_{L+1} \cdot \prod_{j=0}^{l} q_j$ and a ciphertext at the same level is an element in $\mathcal{R}^2_{Q_{l}}$, i.e., a two-component polynomial over $R$ modulo $Q_{l}$.

The CKKS scheme consists of the following algorithms:

\begin{itemize}
    \item $\keygen(1^\lambda, N, [q_0,...,q_L]) \rightarrow (pk, sk, evk):$ On a security parameter $\lambda$, it outputs a public encryption key $pk$, a private decryption key $sk$ and an evaluation key $evk$.

    \item $\encode(z,\Delta) \rightarrow m$: It encodes an (N/2)-dimensional vector of complex numbers $z \in \mathbb{C}^{N/2}$ into a polynomial $m \in \mathcal{R}$ suitable for encryption. During this process, it uses a scaling factor $\Delta$ to control precision.

    \item $\decode(m, \Delta) \rightarrow z$: It is the inverse of \encode where it decodes polynomial $m$ back into a (N/2)-dimensional vector $z$, \emph{approximately} recovering the original values by dividing by $\Delta$.

    \item $\encrypt(m, pk) \rightarrow c$: It encrypts $m$ using the public encryption key $pk$, resulting in a ciphertext $c \in \mathcal{R}^2_{Q_{L}}$.

    \item $\decrypt(c, sk) \rightarrow m$: It recovers $m$ from the ciphertext $c$ using the private decryption key $sk$. 

    \item $\add(c_1, c_2) \rightarrow c_{add}$: It performs homomorphic addition on two ciphertexts $c_1, c_2$, resulting in $c_{add}$. Conceptually, $c_{add}$ is equivalent to $\encrypt(m_1+m_2, pk)$.

    \item $\mult(c_1, c_2, evk) \rightarrow c_{mult}$: Given the evaluation key $evk$ and current multiplicative level $l$ with ciphertext modulus $Q_l$, it performs homomorphic multiplication on two ciphertexts $c_1, c_2 \in \mathcal{R}^2_{Q_l}$, resulting in $c_{mult}$, where $c_{mult} \equiv \encrypt(m_1\cdot{}m_2, pk)$.  
    After multiplication, it rescales the ciphertext modulus reducing it to $Q_{l-1}$, i.e., resulting in $c_{mult} \in \mathcal{R}^2_{Q_{l-1}}$.

    \item $\multconst(c, a, evk) \rightarrow c'$: It is a special case of $\mult$ where a ciphertext $c$ is multiplied by a constant $a \in \mathbb{R}$ producing a new ciphertext $c'$ that is equivalent to $\encrypt(a\cdot{}m, pk)$. Rescaling is also commonly applied, making $c' \in \mathcal{R}^2_{Q_{l-1}}$ as a result of this operation.
    
\end{itemize}


\section{\acro: Framework for PFL Integration with CKKS}
\label{sec:framework}

The goal of this work is to investigate the performance impact of introducing the CKKS scheme to enhance privacy in PFL.
To this end, we first present the system and adversary models under our envisioned privacy-preserving settings in Section~\ref{subsec:system} and Section~\ref{subsec:threat}, respectively.
Next, in Section~\ref{subsec:acro}, we describe the generic composition of PFL and CKKS that forms the basis of our study.
Finally, in Section~\ref{subsec:analysis}, we evaluate the security and soundness of this composition and analyze the CKKS parameter choices required to achieve these properties.

\subsection{System Model}\label{subsec:system}

In this work, we target cross-silo PFL settings as depicted in Figure~\ref{fig:setting}, where each client represents a company or organization (e.g., a bank or hospital). 
In this setting, every client has its own private dataset;
the dataset is not directly shareable due to privacy and legal concerns. 
Instead, the goal is for clients to collaboratively improve their local models by leveraging knowledge from other clients' data without exposing or exchanging the raw data itself.
As a cross-silo setting, the number of clients is relatively small (typically between 2 and 100), and they are expected to participate reliably and faithfully throughout the training process. 
The central server, distinct from the clients, coordinates and orchestrates the training until completion.

\subsection{Threat Model}\label{subsec:threat}
We follow the standard adversary assumption from prior work applying HE to traditional FL~\cite{zhang2020batchcrypt,correia2025federatedlearningapproachhybrid}:  an \emph{honest-but-curious server} and \emph{honest clients}.
Specifically, the server complies with the protocol but may attempt to exploit the information exchanged during protocol execution to violate clients’ privacy, e.g., by launching membership inference attacks using client's local updates.
In line with the cross-silo setting, the clients are organizations subject to privacy regulations and therefore have little incentive to behave maliciously. Accordingly, we consider all clients to be honest throughout the protocol and do not collude with the server.

We assume secure communications between the server and each client, e.g., by employing TLS-based connections, ensuring that no external adversaries (or other clients) can eavesdrop on the data exchanged during the training process. 
Before the training process, we assume the presence of a leader (either a client elected by the server~\cite{zhang2020batchcrypt} or a trusted third party~\cite{correia2025federatedlearningapproachhybrid}) responsible for running \keygen to generate the CKKS keys: $pk$, $sk$, and $evk$. 
The leader then securely distributes them to all clients, while providing only $pk$ and $evk$ to the server.

\subsection{Our Framework: \acro}\label{subsec:acro}

\begin{figure}[h]
    \centering
    \includegraphics[width=\columnwidth]{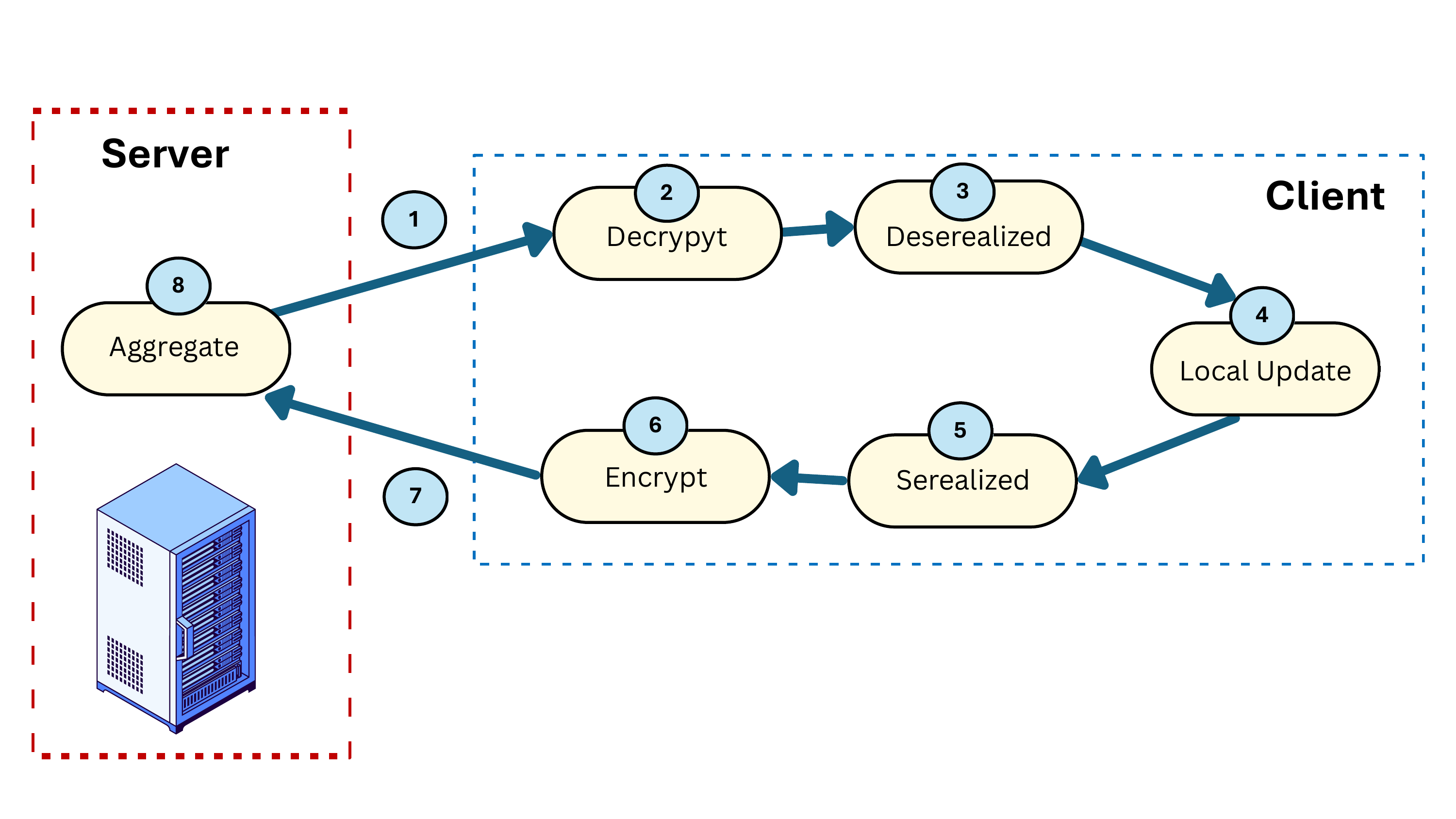}
    \caption{The workflow of \acro}
    \label{fig:workflow}
\end{figure}

\SetKwComment{Comment}{\newline\textcolor{blue}{// }}{}
\renewcommand\CommentSty[1]{\textcolor{blue}{\textit{#1}}} 
\SetKwInOut{Input}{input}
\SetKwInOut{Output}{output}

\begin{algorithm}
    \SetAlgoLined
    \small
    \caption{\acro}\label{alg:acro}    \Input{Number of rounds $t$, CKKS parameters: $\Delta$, $L$, $(q_0,...,q_{L+1})$}
    \Output{$\theta_i$ for client $C_i$}
    \Comment{Setup}
    The server elects a leader to run $\keygen(1^\lambda, (q_0,...,q_L))$\;
    The leader produces $(pk,sk,evk)$, shares them with  clients and sends only $pk,evk$ to server\;
    \Comment{Server}
    The server initializes global model $\Theta$ and uses $pk$ to serialize and encrypt $\Theta$, resulting in $\Theta^{e}$\;
    \For{$round \in \{1,...,t\}$}{
        Send $\Theta^{e}$ to clients\;
        $(\Theta_1^{e}, |D_1|), ..., (\Theta_N^{e}, |D_n|) \leftarrow$ Receive updates from clients\;
        $\Theta^{e} \leftarrow \encrypt(\encode(0,\Delta), pk)$\;
        \For{$j \in \{1,...,n\}$}{
            $w_j \leftarrow n_j/\sum_i^n|D_i|$\;
            $m_j \leftarrow \multconst(\Theta_j^{e}, w_j)$\;
            $\Theta^{e} \leftarrow \add(\Theta^{e},m_j)$\;
        }
    }
    \Comment{Client $C_i$}
    Each client $C_i$ initializes a local model $\theta_i$\;
    \For{$t \in \{1,...,T\}$}{
        $\Theta^{e} \leftarrow$ Receive encrypted global parameters from server\;
        $\Theta \leftarrow \decode(\decrypt(\Theta^{e},sk), \Delta)$\;
        $\Theta \leftarrow Deserialize(\Theta)$\;
        $\hat{\Theta}_i, \hat{\theta}_i \leftarrow \localupdate(\Theta, \theta_i, D_i)$\;
        $\hat{\Theta}_i \leftarrow Serialize(\hat{\Theta}_i)$\;
        $\hat{\Theta}_i^{e} \leftarrow \encrypt(\encode(\hat{\Theta}_i,\Delta), pk)$\;
        Send ($\hat{\Theta}^{e}$, $|D_i|$) to server\;
    }
    Use $\theta_i$ as the final model\;
\end{algorithm}

We present \acro, a generic framework that integrates CKKS into PFL to provide privacy protection against an honest-but-curious server.
\acro operates in two main phases: setup and training.
In the setup phase, following Section~\ref{subsec:threat}, a leader is selected to generate $pk$, $sk$, and $evk$. The leader then distributes $pk$, $sk$, and $evk$ to all clients, while only $pk$ and $evk$ are shared with the server.

The training phase consists of eight steps, illustrated in Figure~\ref{fig:workflow} and detailed in Algorithm~\ref{alg:acro}.
Each client first initializes its personalized model $\theta_i$, while the server initializes the global model $\Theta$ using either random weights or a pre-trained model. Using $pk$, the server encrypt $\Theta$, obtaining $\Theta^e$. 
It then distributes $\Theta^e$ to all clients to begin the first round of PFL training in Step (1).
In Step (2), upon receiving $\Theta^e$, each client $C_i$ decrypts and decodes it to recover the global model as a floating-point vector.
Next, it deserializes the vector into the correct model format (e.g., PyTorch), yielding $\Theta$ in Step (3).

In Step (4), it performs a local PFL update using the underlying PFL algorithm (e.g., \finetune.\localupdate, \ditto.\localupdate or \fedper.\localupdate), resulting in an updated global model $\hat{\Theta}_i$ and a new personalized model $\hat{\theta}_i$. 
$\hat{\theta}_i$ is retained for the next training round while the client prepares $\hat{\Theta}_i$ for transmission to the server.
In Step (5), $\hat{\Theta}_i$ is serialized into a floating-point vector, which is then encoded and encrypted in Step (6);
we note that the CKKS encryption in \acro is performed with the packing method, where it packs up to $N/2$ floating-point values into a single ciphertext.
In Step (7), the client sends the ciphertexts, along with the sample size used for training, back to the server.

Finally, in Step (8), once the server receives ciphertexts from all clients, it performs homomorphic weighted-average aggregation (Lines 7-14 of Algorithm~\ref{alg:acro}), obtaining the aggregate ciphertext $\Theta^e$.
At this point, $\Theta^e$ corresponds to the encrypted version of $\Theta$ as defined in Equation~\ref{eq:weighted-avg}. 
The server then continues with the next PFL round by sending $\Theta^e$ to the clients. 
This process repeats until all rounds ($T$) are consumed, after which each client $C_i$ adopts its $\hat{\theta}_i$ as the final personalized model.


\subsection{Analysis}\label{subsec:analysis}

Here, we analyze security and soundness of \acro.


\textbf{Security.} 
We informally argue that \acro protects the confidentiality of each client's uploaded model update $\hat{\Theta}_i$ against the honest-but-curious server, while providing the formal proof to Appendix~\ref{apdx:proof}.
During training, each client $C_i$ sends only its encrypted update $\hat{\Theta}_i^e$ and dataset size $|D_i|$ to the server; see Line~22 of Algorithm~\ref{alg:acro}. 
The dataset size is intentionally revealed because it is required for weighted aggregation, but it discloses only the number of local samples and not the raw training records.

Let $\mathcal{L}$ denote the information explicitly visible to the server, including the number of clients, client participation in each round, ciphertext sizes, model/update dimensions, public CKKS parameters, and $|D_i|$.
Given $\mathcal{L}$, the server does not learn $\hat{\Theta}_i$ from $\hat{\Theta}_i^e$.
This follows from the IND-CPA security of CKKS under the RLWE hardness assumption~\cite{cheon2017homomorphic,cheon2018full,benaissa2021tenseal}.
In particular, the server receives only the public key and evaluation key, but not the secret key.
Therefore, an honest-but-curious server cannot distinguish encryptions of two equal-length candidate updates, except with negligible probability.
As a result, \acro provides privacy protection of training data w.r.t. the threat model described in Section~\ref{subsec:threat}.


\textbf{Soundness.} Under ideal CKKS parameters, the aggregation performed by the server produces an encrypted approximation of the aggregate model $\Theta$ defined in Equation~\ref{eq:weighted-avg}. 
Clients can therefore recover an approximation of $\Theta$ as in a standard PFL setup, i.e., the one without HE. 
However, because CKKS introduces noise during encryption, decrypted results may deviate slightly from the exact plaintext values. The extent of this deviation depends on the choice of CKKS parameters in \acro, which will be explored next.

\textbf{CKKS parameters.}
In CKKS, 4 parameters are critical to decryption precision and thus the soundness of \acro:
$N$, $L$, $(q_0, \ldots, q_{L+1})$, and $\Delta$.
The parameter $N$ also directly determines the security level of CKKS. 
To achieve the standard 128-bit security level, we follow the recommended setting of $N = 8192$~\cite{albrecht2022homomorphic,benaissa2021tenseal,furka2023guidelines}.
The multiplicative depth $L$ specifies how many homomorphic multiplications can be supported. In \acro, we only need one multiplication (Line 9 of Algorithm~\ref{alg:acro}), so $L = 1$. 
This results in three ciphertext moduli: $(q_0, q_1, q_2)$, where we refer $q_0$ and $q_2$ as the outer prime and $q_1$ as the inner prime.

Following standard practice~\cite{qiu2022privacy,rahulamathavan2022privacy,pan2024fedshe}, we fix $\Delta = q_1$, $q_0 = q_2$ with $q_0 > q_1$, and require $\log_2(Q_1) = \log_2(q_0) + \log_2(q_1) + \log_2(q_2) < 218$ to maintain 128-bit security~\cite{albrecht2022homomorphic,bossuat2024security,pan2024fedshe}.
As a result, we have the following constraints for selecting CKKS parameters in \acro:

\begin{equation}\label{eq:fixed}
    N=8192,\ L=1,\ q_0=q_2,\ q_1=\Delta
\end{equation}

\begin{equation}\label{eq:ineq}
    q_0>q_1,\ \log_2(q_0) + \log_2(q_1) + \log_2(q_2)<218
\end{equation}

Given these constraints, to use \acro, we only need to select values for the outer prime ($q_0$) and the inner prime ($q_1$) that satisfy the inequality in Equation~\ref{eq:ineq}. Once chosen, the full parameter set can be populated using Equation~\ref{eq:fixed}.

In CKKS, decryption precision is largely influenced by inner and outer primes.
A larger $\Delta$ ($=q_1$) retains more precision in the decimal part, while the decryption precision of the integer part depends on the gap between $q_0$ and $q_1$.
Moreover, $\Delta$ determines how plaintext values are scaled relative to the noise. If $q_1$ (and thus $\Delta$) is too small, the scaling becomes insufficient to reliably distinguish plaintext values from the accumulated noise during decryption.
However, increasing either $q_0$ or $q_1$ enlarges the total modulus $Q_L$, which in turn requires more bits to represent ciphertexts.
This results in larger ciphertext sizes, higher bandwidth consumption and longer computation times.

Consequently, CKKS introduces an inherent trade-off: improving numerical precision comes at the cost of increased computational and communication overheads. 
To balance these competing factors, one must carefully select CKKS parameters that minimize these overheads while maintaining sufficient accuracy for the downstream machine learning task.
In the following sections, we empirically explore this trade-off to identify practical parameter choices. 
To the best of our knowledge, this issue has not been systematically investigated in prior work, including studies of CKKS in traditional FL.

\section{Experiments}
\label{sec:experiments}

\subsection{Setup}

\textbf{Implementation.}
We implement \acro in Python using the Flower federated learning framework~\cite{beutel2020flower}.
Flower provides a standard FL pipeline based on two functions: \texttt{fit} and \texttt{aggregate}.
The \texttt{fit} function is executed by a Flower client and defines the logic for local training on its own dataset. 
It returns the updated model along with the number of trained samples, which Flower automatically forwards to the server.
Upon receiving results from clients, Flower invokes \texttt{aggregate} on the server, which implements the aggregation logic producing an aggregate result. 
This result is then sent back to the clients for the next training round.

To implement \acro in Flower, we modify the \texttt{fit} function to implement Lines 19-24 of Algorithm~\ref{alg:acro} and the server’s \texttt{aggregate} function to implement Lines 7–14.
We use the Tenseal library~\cite{benaissa2021tenseal} to provide Python APIs for CKKS operations in our implementation.

\textbf{Hardware.}
We evaluate our \acro implementation by simulating both clients and server on the same machine: an i7-11700K desktop with 32 GB of RAM and NVIDIA GeForce RTX 4060 GPU. 
In total, all experiments take around $\approx$ 336 GPU hours and $\approx$  720 hours.

\begin{table*}[t]
\centering
\caption{Summary of datasets used in our experiments.}
\label{tab:dataset-statistics}
\begin{tabular}{ccccccc}
\toprule
\textbf{Dataset} & \textbf{Dataset Type} & \textbf{\# Samples} & \textbf{\# Classes} & \textbf{Task} & \textbf{Model} & \textbf{\# Parameters} \\
\midrule

FEMNIST &
Grayscale images &
805,263 &
62 &
Handwritten classification &
2D CNN &
\ignore{
\begin{tabular}[c]{@{}l@{}}
Conv(1$\to$32, 5$\times$5) + ReLU + MaxPool \\
Conv(32$\to$64, 5$\times$5) + ReLU + MaxPool \\
FC(3136$\to$2048) + ReLU \\
FC(2048$\to$62)
\end{tabular} &
}
6.60M \\

CelebA &
RGB images &
202,599 &
2 &
Smiling detection &
2D CNN&
\ignore{
\begin{tabular}[c]{@{}l@{}}
Conv(3$\to$32, 3$\times$3) + ReLU + MaxPool $\times$4 \\
FC(800$\to$2)
\end{tabular} &
}
30.2K \\

Sentiment140 &
Tweet texts &
1,600,498 &
2 &
Sentiment analysis &
1D CNN &
\ignore{
\begin{tabular}[c]{@{}l@{}}
Embedding(5268$\to$20) \\
Conv1D(20$\to$32, k=3) + ReLU + MaxPool \\
Conv1D(32$\to$64, k=3) + ReLU + MaxPool \\
Conv1D(64$\to$128, k=3) + ReLU \\
GlobalMaxPool + FC(128$\to$2)
\end{tabular} &
}
8M \\
\bottomrule
\end{tabular}%
\end{table*}

\textbf{Datasets and Models.} 
We use three datasets from the LEAF federated learning benchmark~\cite{caldas2018leaf} in our experiments: FEMNIST, CelebA, and Sentiment140.
The datasets and the corresponding models used for training are summarized in Table~\ref{tab:dataset-statistics} and described in detail below:

\begin{enumerate}
    \item \emph{FEMNIST} represents an image classification task, consisting of 28×28 grayscale images of handwritten digits and both lowercase and uppercase letters (62 classes in total).
    It contains a total of 805,263 samples.
    To evaluate \acro, we adopt the CNN model implementation provided in LEAF, which consists of five layers: two convolutional layers (with 32 and 64 filters of size 5×5), one fully connected hidden layer with 2048 neurons, and a final output layer with a softmax activation over 62 classes.
    
    \item \emph{CelebA} is a large-scale face attribute classification dataset containing 202,599 RGB images of celebrity faces annotated with 40 binary attributes.
    Following the LEAF benchmark, we focus on the binary classification task for the \textit{smiling} attribute.
    We use the 2D CNN model provided by LEAF. 
    The model has four convolutional layers and one fully connected layer followed by a softmax layer for binary classification. 
    
    \item \emph{Sentiment140} is a LEAF dataset for sentiment classification.
    It contains 1,600,498 tweets collected from 660,120 users, where each tweet is labeled with either positive or negative sentiment.
    The original LEAF implementation uses a logistic regression model for sentiment analysis.
    However, this model is not suitable for \fedper because it consists of only a single layer and therefore cannot be separated into the base and personalized layers required by \fedper.
    Motivated by~\cite{kim2014convolutional}, we instead adopt a 1D CNN model consisting of three convolutional layers with 32, 64, and 128 filters, followed by a fully connected layer that produces the final sentiment prediction.
    \ignore{
    the model used for this dataset is logistic regression,
    which only has one linear layer. Because it only has a
    single layer, it is difficult to divide the model into
    shared layers and personalized layers. However, several
    personalized federated learning (PFL) algorithms require
    the model to be separated in this way.
    Therefore, we use a deeper model based on a 1D CNN.
    This type of model makes it easier to separate the model
    into shared feature extraction layers and personalized
    classification layers for each client.
    First, each tweet is converted into a vector using an
    embedding layer. Then the sequence is passed through
    three 1D convolution layers with 32, 64, and 128 filters.
    All convolution layers use a kernel size of 3 and ReLU
    activation. Max pooling is applied after the first and
    second convolution layers to reduce the sequence size.
    After the third convolution layer, global max pooling is
    used to obtain a fixed-length feature vector. Finally,
    this vector is passed to a fully connected layer with softmax activation to
    produce the final sentiment prediction.
    }
\end{enumerate}

\ignore{
We use two datasets, FEMNIST and CelebA, from the LEAF federated learning benchmark~\cite{caldas2018leaf}. 

\begin{itemize}
    \item FEMNIST represents an image classification task, consisting of 28×28 grayscale images of handwritten digits and both lowercase and uppercase letters (62 classes in total). 
    It contains a total of 805,263 samples. 
    To evaluate \acro, we adopt the CNN implementation provided in LEAF, which consists of XXXX (five?) layers with a total of XXX neurons, with the last layer being a softmax activation layer.

    \item CelebA is a face attribute classification dataset containing 202,599 celebrity images annotated with 40 binary attributes. 
    Following LEAF, we focus on the binary classification task of the ``smiling'' attribute. 
    We use the CNN model provided in LEAF, consisting of XXX layers with a total of XXX neurons, with the final layer being a softmax activation layer.

    \ignore{
    \item Sentiment140 contains 1,600,000 tweets annotated with binary sentiment labels (positive or negative), making it a text classification task.
    Following~\cite{}, we adopt a CNN-based model with five layers and a total of 50 neurons, with the final layer being a softmax activation layer.
    }
\end{itemize}
}

\textbf{PFL Algorithms and Parameters.}
We evaluate \acro on the dataset using three underlying PFL algorithms: \finetune, \fedper and \ditto. 
Specifically, for \fedper, we use the last layer (classification head/fully connected layer+softmax activation) as the personalized layer, while the remaining are the base layers shared with the server during training. 
In \ditto, we set the hyperparameter $\lambda$ (see Section~\ref{subsec:pfl}) to 0.1, following previous work~\cite{li2021ditto}.

For all algorithms, we use the following setup: a learning rate of 0.01, 20 communication rounds, 3 PFL clients (for cross-silo settings), the entire dataset divided into 100 partitions where each client holds one partition, and an 80/20 train-test split. 
For the FEMNIST dataset, data is partitioned in a non-IID manner using the Dirichlet partitioning method with $\alpha = 0.1$.
For the CelebA and Sentiment140 datasets, the tasks are binary classification, which makes it difficult to reliably construct non-IID partitions based on prediction labels.
Therefore, we partition these two datasets in an IID manner instead.

\textbf{Baseline and Metrics.} 
In this work, we consider the baseline to be the respective unmodified (unencrypted) PFL algorithm. 
Our goal is to explore \acro trade-off between communication/computational overhead and precision. Accordingly, we report three performance metrics in our experiments:

\begin{itemize}
    \item \textbf{Communication} (Section~\ref{subsec:communication-overhead}) measured by the total number of bytes each client transmits and receives in each training round.
    \item \textbf{Computation} (Section~\ref{subsec:computation-overhead}) as the total time (CPU + GPU) performed on each client to complete each training round.
    \item \textbf{Precision}(Section~\ref{subsec:precision-overhead}), where we measure at two levels: parameter and model level.
\end{itemize}

\ignore{
at three levels of granularity:

\begin{itemize}
    \item \textbf{Parameter-level:} the average L1 error between the final personalized parameters from \acro and those from the baseline across all clients.

    \item \textbf{Loss-level:} the average L1 error between the final training loss in \acro and that of the baseline across all clients.

    \item \textbf{Accuracy-level:} the average L1 error of test accuracy between \acro and the baseline across all clients.
\end{itemize}

This allows us to quantify the impact of precision loss introduced by the CKKS scheme at different stages of the learning process.
}

\subsection{Communication}\label{subsec:communication-overhead}
\ignore{
\begin{table*}[t]
\caption{Experiment Setup}
\label{tab:experiment-table}
\centering
\begin{tabular}{|c|c|c|c|}
\hline
Experiment & Parameters Being Tested & Range/Setup & Metrics Measured \\ \hline

Exp 1.1 &
\begin{tabular}[c]{@{}c@{}}Middle prime\end{tabular} &
18--40 &
Accuracy, Runtime, Bandwidth \\ \hline

Exp 1.2 &
\begin{tabular}[c]{@{}c@{}}Outer prime\end{tabular} &
18--40 &
Runtime, Bandwidth \\ \hline

Exp 1.3 &
Float precision &
Float32&
Accuracy, Runtime \\ \hline

Exp 2 &
Performance on HE vs non-HE &
HE enabled vs HE disabled &
Accuracy, Bandwidth, Runtime \\ \hline

\end{tabular}%
\end{table*}
}

\begin{figure*}[t]
\centering

\begin{subfigure}[b]{0.3\textwidth}
    \centering
    \includegraphics[width=\textwidth]{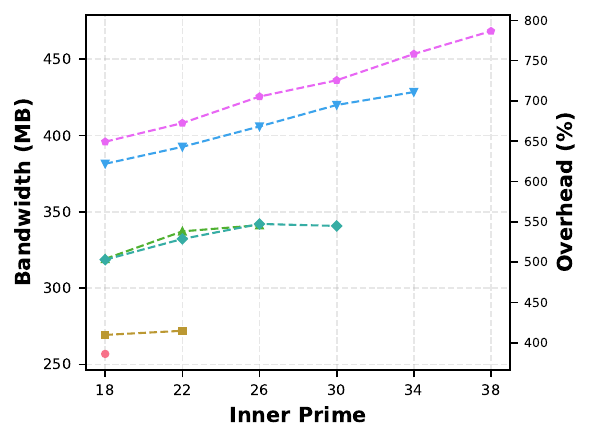}
    \caption{\finetune\ -- FEMNIST}
    \label{fig:ft_bandwidth_femnist}
\end{subfigure}
\hfill
\begin{subfigure}[b]{0.3\textwidth}
    \centering
    \includegraphics[width=\textwidth]{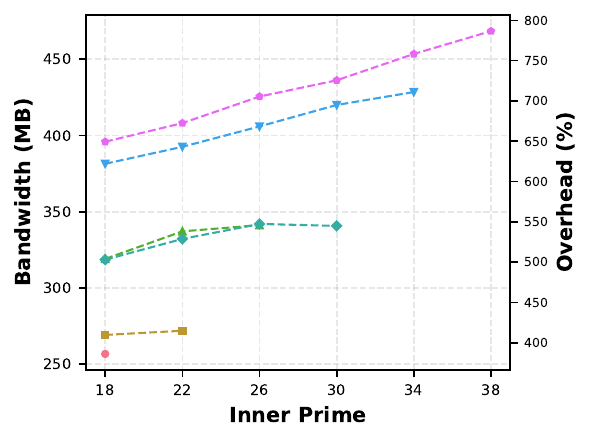}
    \caption{\ditto\ -- FEMNIST}
    \label{fig:ditto_bandwidth_femnist}
\end{subfigure}
\hfill
\begin{subfigure}[b]{0.3\textwidth}
    \centering
    \includegraphics[width=\textwidth]{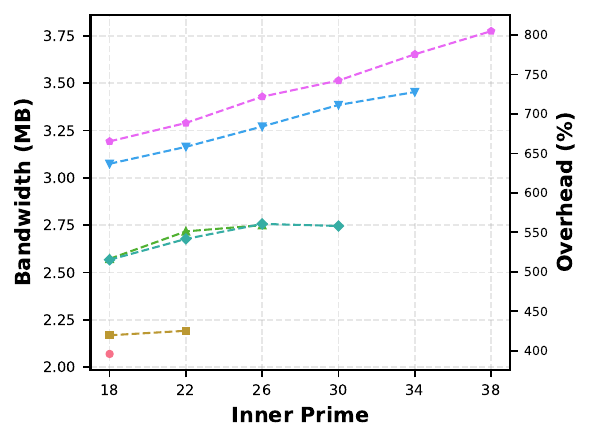}
    \caption{\fedper\ -- FEMNIST}
    \label{fig:fedper_bandwidth_femnist}
\end{subfigure}

\vspace{0.5cm}

\begin{subfigure}[b]{0.3\textwidth}
    \centering
    \includegraphics[width=\textwidth]{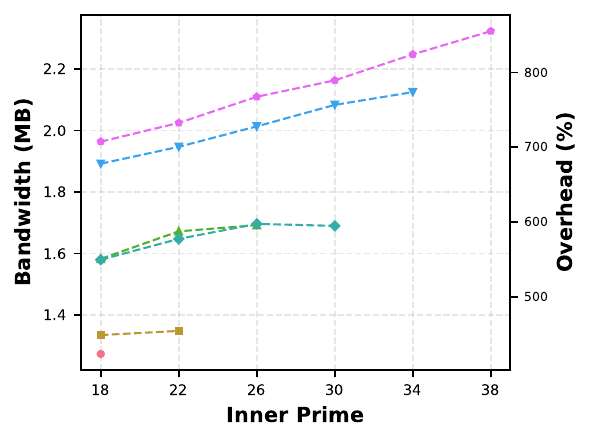}
    \caption{\finetune\ -- CelebA}
    \label{fig:ft_bandwidth_celebA}
\end{subfigure}
\hfill
\begin{subfigure}[b]{0.3\textwidth}
    \centering
    \includegraphics[width=\textwidth]{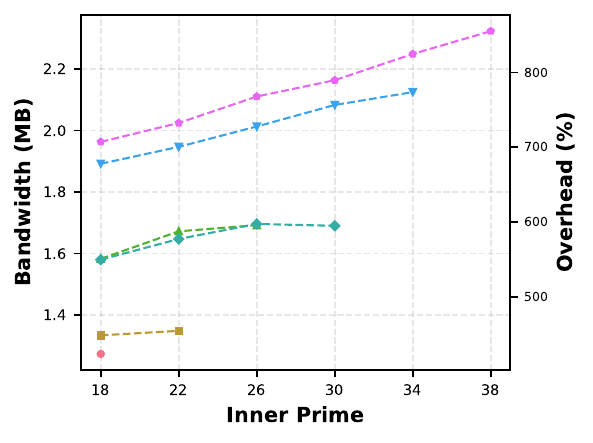}
    \caption{\ditto\ -- CelebA}
    \label{fig:ditto_bandwidth_celebA}
\end{subfigure}
\hfill
\begin{subfigure}[b]{0.3\textwidth}
    \centering
    \includegraphics[width=\textwidth]{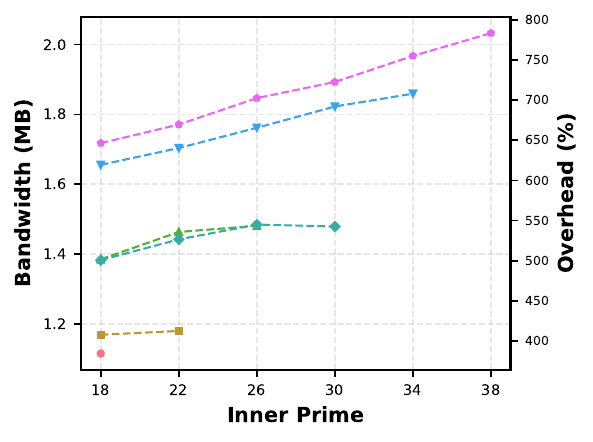}
    \caption{\fedper\ -- CelebA}
    \label{fig:fedper_bandwidth_celebA}
\end{subfigure}

\vspace{0.5cm}

\begin{subfigure}[b]{0.3\textwidth}
    \centering
    \includegraphics[width=\textwidth]{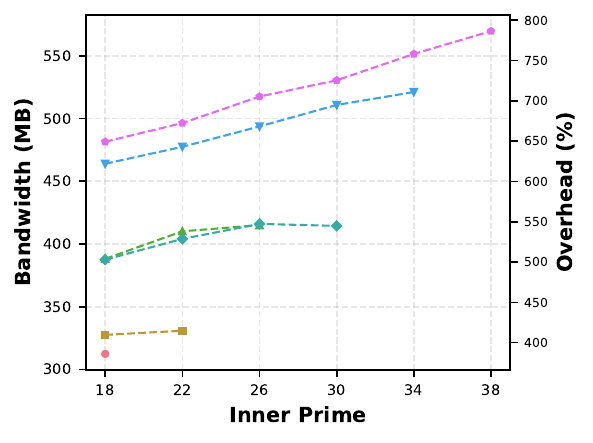}
    \caption{\finetune\ -- Sentiment140}
    \label{fig:ft_bandwidth_sentiment}
\end{subfigure}
\hfill
\begin{subfigure}[b]{0.3\textwidth}
    \centering
    \includegraphics[width=\textwidth]{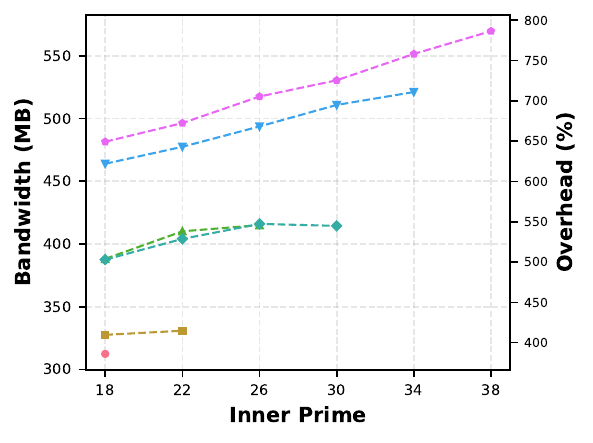}
    \caption{\ditto\ -- Sentiment140}
    \label{fig:ditto_bandwidth_sentiment}
\end{subfigure}
\hfill
\begin{subfigure}[b]{0.3\textwidth}
    \centering
    \includegraphics[width=\textwidth]{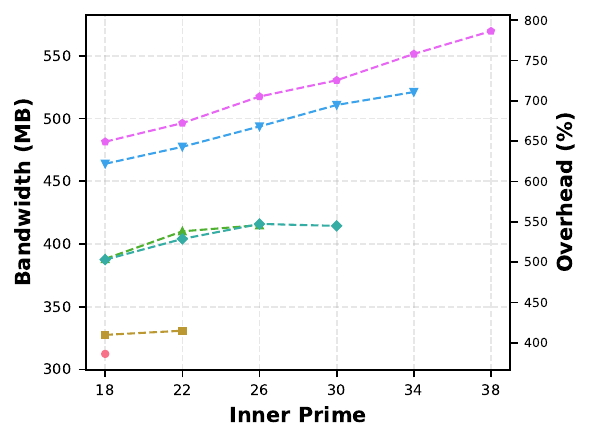}
    \caption{\fedper\ -- Sentiment140}
    \label{fig:fedper_bandwidth_sentiment}
\end{subfigure}

\caption{\acro's bandwidth usage (absolute value on the left y-axis) and its overhead w.r.t. the baseline (percentage on the right y-axis). 
Each line corresponds to a different \emph{outer} prime value:
\textcolor{outer20}{\textbf{$\bullet$=20}},
\textcolor{outer24}{\textbf{$\blacksquare$=24}},
\textcolor{outer28}{\textbf{$\blacktriangle$=28}},
\textcolor{outer32}{\textbf{$\blacklozenge$=32}},
\textcolor{outer36}{\textbf{$\blacktriangledown$=36}},
\textcolor{outer40}{\textbf{\blackpentagon=40}};  
\emph{inner} prime bits are varied along in the x-axis.
Results are averaged over 3 seeds; standard deviation is $<0.01$ and omitted.
}
\label{fig:bandwidth_trends}
\end{figure*}

Figure~\ref{fig:bandwidth_trends} illustrates the communication cost of \acro and its overhead relative to the baseline. 
The percentage communication overhead exhibits a consistent trend across all datasets and PFL algorithms: both inner and outer primes directly influence the overhead, with larger values leading to higher communication cost.
For example, with a fixed 40-bit outer prime, increasing the inner prime from $18$ to $38$ results in at least an additional 100\% communication overhead. 
Also, for a fixed inner prime of $22$, increasing the outer prime from the lowest value ($24$) to the largest value ($40$) leads to an increase of at least 250\%.

The overall communication cost follows a similar trend to the \% overhead, where the choice of inner/outer prime bits influence actual bandwidth usage.
For FEMNIST, \finetune and \ditto exhibit similar bandwidth consumption ($\approx$257--468\,MB), whereas \fedper reduces the communication cost to less than 4\,MB. 
This reduction is expected, as \fedper decouples personalized parameters from shared parameters and only transmits a smaller subset of shared parameters (i.e., the base layer of the 2D CNN used in FEMNIST which contributes to 79\% of the overall model parameters), thereby significantly lowering the amount of encrypted data exchanged.

A similar trend is observed in Sentiment140, where \finetune and \ditto again consume comparable bandwidth to their FEMNIST counterparts. However, \fedper incurs higher bandwidth usage in this case because the shared component corresponds to the base layers of a 1D CNN, which constitutes a larger fraction ($\approx$ 99.98\%) of the total model parameters.

Finally, due to its smaller model size, CelebA results in the lowest bandwidth consumption (less than 2.5\,MB across all settings). However, when compared to the baseline, its relative overhead remains comparable to that of other datasets.

\begin{Takeaway}{Impact on Communication}
xIn \acro, the choice of both inner and outer CKKS primes directly affects communication overhead  over the unencrypted baseline, with larger values leading to higher overhead. In terms of absolute cost, the choice of PFL algorithm and model architecture further influences the overall bandwidth usage.
\end{Takeaway}

%

\begin{figure*}[t]
\centering

\begin{subfigure}[b]{0.3\textwidth}
    \centering
    \includegraphics[width=\textwidth]{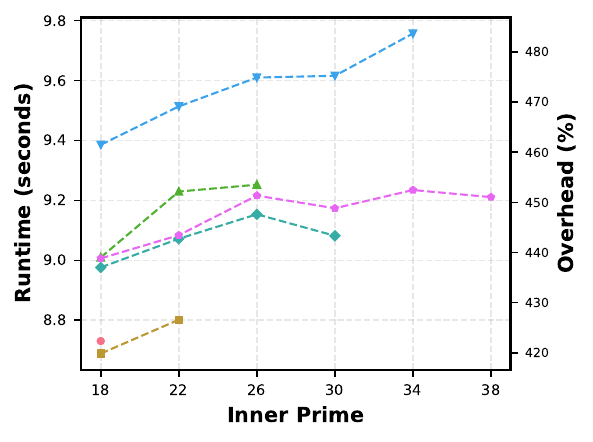}
    \caption{\finetune\ -- FEMNIST}
    \label{fig:ft_runtime_femnist}
\end{subfigure}
\hfill
\begin{subfigure}[b]{0.3\textwidth}
    \centering
    \includegraphics[width=\textwidth]{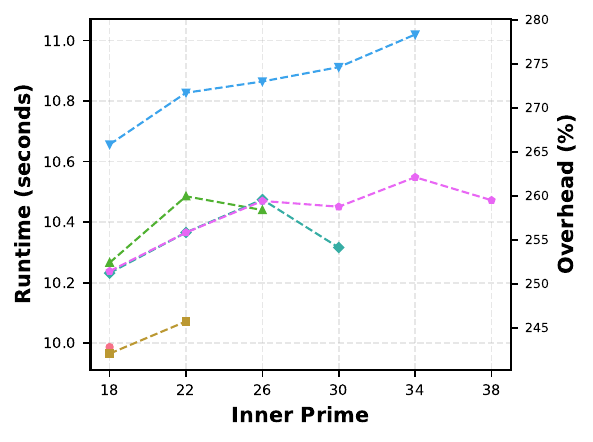}
    \caption{\ditto\ -- FEMNIST}
    \label{fig:ditto_runtime_femnist}
\end{subfigure}
\hfill
\begin{subfigure}[b]{0.3\textwidth}
    \centering
    \includegraphics[width=\textwidth]{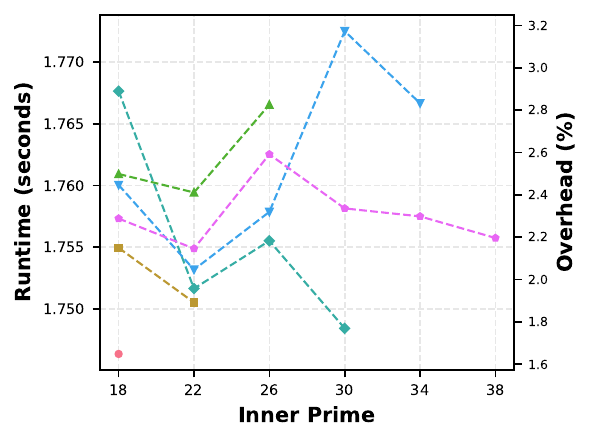}
    \caption{\fedper\ -- FEMNIST}
    \label{fig:fedper_runtime_femnist}
\end{subfigure}

\vspace{0.5cm}
\begin{subfigure}[b]{0.3\textwidth}
    \centering
    \includegraphics[width=\textwidth]{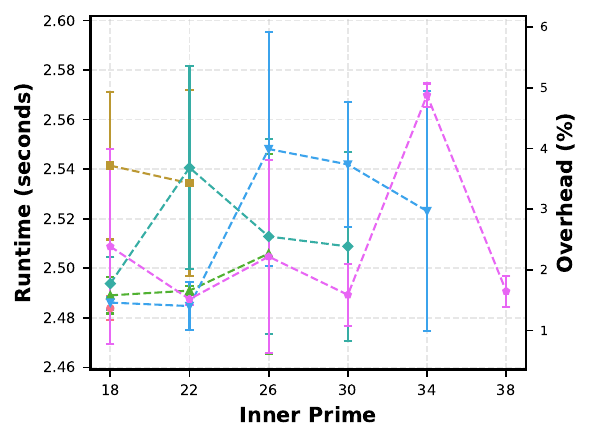}
    \caption{\finetune\ -- CelebA}
    \label{fig:ft_runtime_celebA}
\end{subfigure}
\hfill
\begin{subfigure}[b]{0.3\textwidth}
    \centering
    \includegraphics[width=\textwidth]{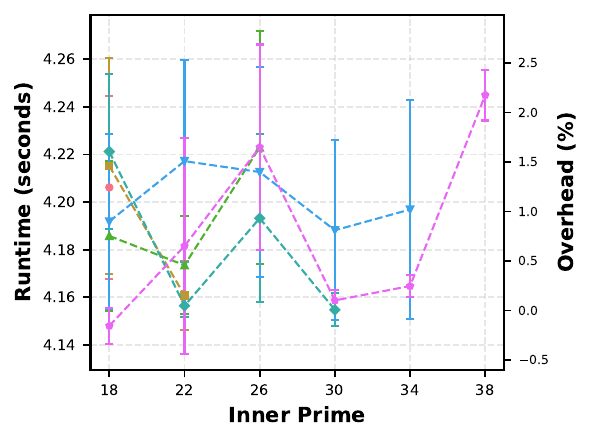}
    \caption{\ditto\ -- CelebA}
    \label{fig:ditto_runtime_celebA}
\end{subfigure}
\hfill
\begin{subfigure}[b]{0.3\textwidth}
    \centering
    \includegraphics[width=\textwidth]{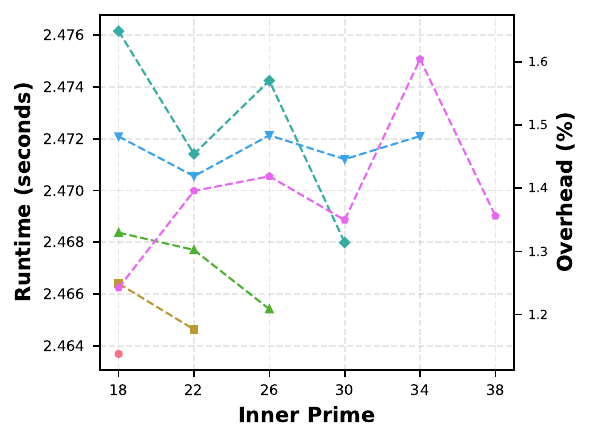}
    \caption{\fedper\ -- CelebA}
    \label{fig:fedper_runtime_celebA}
\end{subfigure}

\vspace{0.5cm}
\begin{subfigure}[b]{0.3\textwidth}
    \centering
    \includegraphics[width=\textwidth]{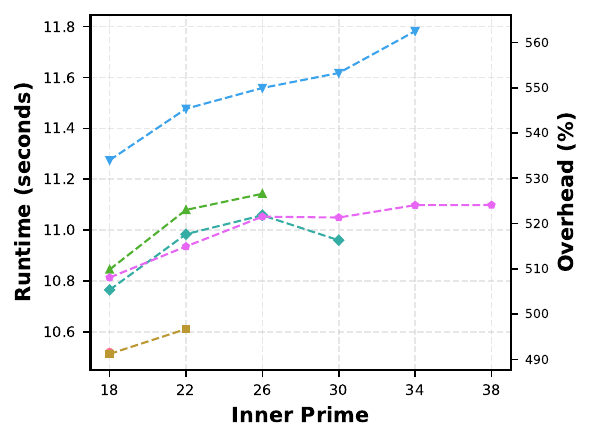}
    \caption{\finetune\ -- Sentiment140}
    \label{fig:ft_runtime_sentiment}
\end{subfigure}
\hfill
\begin{subfigure}[b]{0.3\textwidth}
    \centering
    \includegraphics[width=\textwidth]{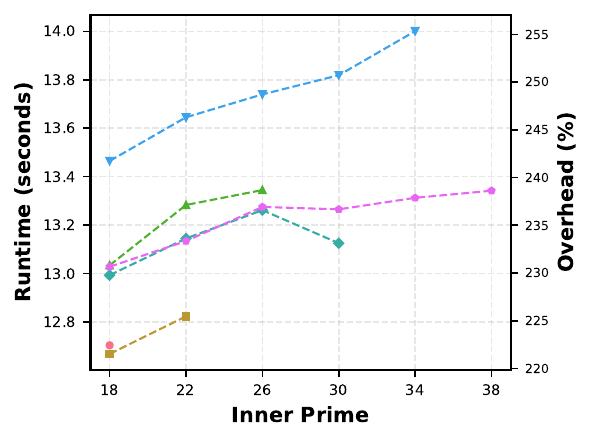}
    \caption{\ditto\ -- Sentiment140}
    \label{fig:ditto_runtime_sentiment}
\end{subfigure}
\hfill
\begin{subfigure}[b]{0.3\textwidth}
    \centering
    \includegraphics[width=\textwidth]{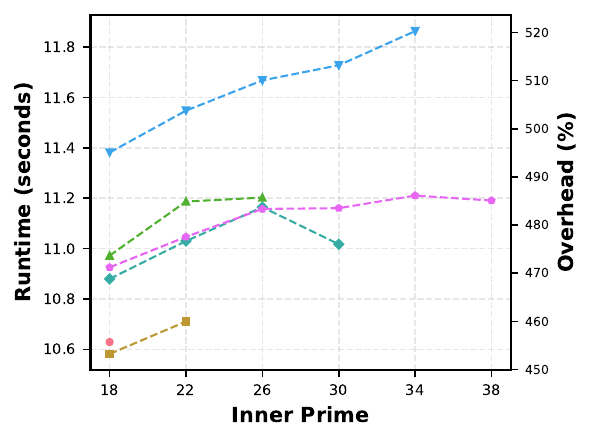}
    \caption{\fedper\ -- Sentiment140}
    \label{fig:fedper_runtime_sentiment}
\end{subfigure}

\caption{Runtime usage of \acro\ with various inner and outer prime values (in bits). Each line corresponds to a different \emph{outer} prime value:
\textcolor{outer20}{\textbf{$\bullet$=20}},
\textcolor{outer24}{\textbf{$\blacksquare$=24}},
\textcolor{outer28}{\textbf{$\blacktriangle$=28}},
\textcolor{outer32}{\textbf{$\blacklozenge$=32}},
\textcolor{outer36}{\textbf{$\blacktriangledown$=36}},
\textcolor{outer40}{\textbf{\blackpentagon=40}}; 
}
\label{fig:runtime}
\end{figure*}

\begin{table*}[t]
\centering
\begin{tabular}{llrrrrr}
\toprule
\textbf{Algorithm} & \textbf{Dataset}
  & \textbf{(i) Decrypt (ms)}
  & \textbf{(ii) Training (ms)}
  & \textbf{(iii) Encrypt (ms)}
  & \textbf{CKKS Total (ms)}
  & \textbf{Total Fit (ms)} \\
\midrule
\multirow{3}{*}{\finetune}
  & FEMNIST    & 899.6$\pm$1.6 (9.8\%)   & 1259.7$\pm$4.1 (13.7\%)  & 6645.5$\pm$27.1 (72.2\%) & 7545.1$\pm$27.1 (82.0\%) & 9204.3$\pm$28.7  \\
  & \cellcolor{gray!20}CelebA     & \cellcolor{gray!20}5.6$\pm$0.1 (0.22\%)    & \cellcolor{gray!20}2036.7$\pm$5.7 (81.7\%)  & \cellcolor{gray!20}33.8$\pm$0.7 (1.36\%)    & \cellcolor{gray!20}39.4$\pm$0.8 (1.58\%)    & \cellcolor{gray!20}2490.7$\pm$6.2   \\
  & Sentiment  & 1118.2$\pm$9.6 (10.1\%) & 1051.4$\pm$7.1 (9.47\%)  & 8212.2$\pm$3.4 (74.0\%)  & 9330.5$\pm$11.5 (84.1\%) & 11098.1$\pm$4.7  \\
\midrule
\multirow{3}{*}{\ditto}
  & FEMNIST    & 908.7$\pm$7.7 (8.7\%)    & 2499.6$\pm$13.2 (23.9\%) & 6657.8$\pm$17.4 (63.6\%) & 7566.6$\pm$25.0 (72.3\%) & 10472.5$\pm$32.0 \\
  & \cellcolor{gray!20}CelebA     & \cellcolor{gray!20}5.5$\pm$0.0 (0.13\%)     & \cellcolor{gray!20}3720.2$\pm$12.0 (87.6\%) & \cellcolor{gray!20}32.8$\pm$0.3 (0.77\%)    & \cellcolor{gray!20}38.3$\pm$0.3 (0.90\%)    & \cellcolor{gray!20}4244.9$\pm$10.7  \\
  & Sentiment  & 1118.9$\pm$10.2 (8.4\%)  & 3278.6$\pm$16.1 (24.6\%) & 8252.0$\pm$12.0 (61.8\%) & 9370.9$\pm$10.1 (70.3\%) & 13342.0$\pm$12.1 \\
\midrule
\multirow{3}{*}{\fedper}
  & \cellcolor{gray!20}FEMNIST    & \cellcolor{gray!20}7.8$\pm$0.0 (0.44\%)     & \cellcolor{gray!20}1263.9$\pm$2.8 (71.6\%)  & \cellcolor{gray!20}52.8$\pm$0.8 (3.0\%)     & \cellcolor{gray!20}60.7$\pm$0.8 (3.4\%)     & \cellcolor{gray!20}1764.4$\pm$13.2  \\
  & \cellcolor{gray!20}CelebA     & \cellcolor{gray!20}4.6$\pm$0.0 (0.19\%)     & \cellcolor{gray!20}2024.2$\pm$3.5 (82.0\%)  & \cellcolor{gray!20}30.0$\pm$0.4 (1.2\%)     & \cellcolor{gray!20}34.7$\pm$0.4 (1.41\%)    & \cellcolor{gray!20}2469.3$\pm$3.1   \\
  & Sentiment  & 1123.5$\pm$5.9 (10.6\%)  & 1199.7$\pm$2.9 (11.5\%)  & 8156.4$\pm$12.3 (77.8\%) & 9279.9$\pm$7.6 (88.5\%)  & 11191.0$\pm$6.6  \\
\bottomrule
\end{tabular}%
\caption{Computation time breakdown per federated round using CKKS homomorphic encryption (prime config [40,38,40], mean$\pm$std across 3 seeds). All times in milliseconds (ms). (i) CKKS decryption, (ii) local training, (iii) CKKS encryption. \colorbox{gray!20}{Gray} rows indicate settings where training (ii) dominates over CKKS operations (i)+(iii), causing no clear trend.}
\label{tab:computation_time}
\end{table*}
\subsection{Computation}\label{subsec:computation-overhead}
We define the computation cost of \acro as the per-round execution time of Algorithm~\ref{alg:acro} on each PFL client. Figure~\ref{fig:runtime} reports both the absolute runtime and the corresponding overhead (in \%) relative to the baseline under varying \acro's CKKS inner and outer prime bit sizes.
We observe that only five out of the nine settings exhibit a clear trend: similar to communication overhead, increasing the inner and outer prime values leads to higher computation time. These settings include all three PFL algorithms on the Sentiment140 dataset (Figure~\ref{fig:ft_runtime_sentiment}--\ref{fig:fedper_runtime_sentiment}) and two PFL algorithms on the FEMNIST dataset (Figure~\ref{fig:ft_runtime_femnist} and Figure~\ref{fig:ditto_runtime_femnist}).

Among these five settings, the two that employ \ditto incur a computation overhead of approximately $2.2$-$2.8\times$ over the baseline, while the remaining three exhibit higher overheads of $4.2$--$5.6\times$. 
The remaining settings do not show a clear relationship between computation cost and the CKKS prime parameters.

For settings that exhibit a clear trend, each PFL round takes around 10 seconds, with low variance (standard deviation $< 0.01$). 
In contrast, settings without a clear trend complete significantly faster (2-4 seconds) and may exhibit higher runtime variance; however, the absolute difference between the maximum and minimum runtime remains small, typically within the first or second decimal place (e.g., 2.464 vs.\ 2.476 seconds in Figure~\ref{fig:fedper_runtime_celebA}).

To better understand this behavior, we perform a runtime microbenchmark that breaks the computation cost into finer-grained operations. 
Based on Lines~14-24 of Algorithm~\ref{alg:acro}, we consider three main components: (i) CKKS decryption (Lines~17-18), (ii) local training (Line~19), and (iii) CKKS encryption (Lines~20-21). 
Among these, (i) and (iii) constitute the sources of \acro computation overhead compared to the unencrypted baseline.

Table~\ref{tab:computation_time} presents the runtime breakdown. For the four settings that do not exhibit a clear trend in Figure~\ref{fig:runtime}, the total runtime is dominated by (ii) local training, which varies significantly across different runs/seeds. 
In contrast, the runtime of (i) and (iii) is much smaller (0.90-3.4\% of the overall runtime) and stable. 
As a result, incorporating \acro (i.e., adding (i) and (iii)) does not significantly affect the overall computation time, with the overhead remaining below 3.4\%.
Notably, three of these four settings correspond to CelebA, where the model size is relatively small, leading to faster execution of (i)/(iii) compared to (ii). 
The remaining setting corresponds to the \fedper--FEMNIST combination. In this case, the base layers of FEMNIST constitute only about 3.4\% of the total model parameters; thus, under \fedper, only this small subset is encrypted and decrypted, resulting in minimal overhead from (i) and (iii).

From Table~\ref{tab:computation_time}, the remaining five settings that exhibit a clear relationship between CKKS inner/outer prime sizes and computation cost are dominated by CKKS encryption (iii) and CKKS decryption (i), with encryption contributing the larger share of the overhead around $7.3$-$7.4\times$, which is consistent with the finding from prior work~\cite{jiang2022fhebench}.

\begin{table}[t]
\centering
\caption{Encryption and serialization time (ms) across outer prime bits, measured on the FEMNIST model.}
\label{tab:ckks_enc_time}
\begin{tabular}{cc}
\toprule
\textbf{Outer Prime (bits)} & \textbf{Encryption + Serialization Time (ms)} \\
\midrule
\rowcolor{gray!25} 20 & 6531.2 \\
\rowcolor{gray!25} 22 & 6547.0 \\
24 & 6501.0 \\
\rowcolor{gray!25} 26 & 6664.7 \\
\rowcolor{gray!25} 28 & 6742.3 \\
\rowcolor{gray!25} 30 & 6737.0 \\
32 & 6737.9 \\
\rowcolor{gray!25} 34 & 6979.2 \\
\rowcolor{gray!25} 36 & 7131.1 \\
\rowcolor{gray!25} 38 & 7067.9 \\
40 & 6731.8 \\
\bottomrule
\end{tabular}
\end{table}
We also observe a non-monotonic behavior in computation time w.r.t. the outer prime size. In particular, certain configurations (e.g., outer prime = 36 bits) exhibit slightly higher runtime compared to larger sizes (e.g., outer prime = 40 bits), which appears counterintuitive.
Upon further inspection, we found that this behavior stems from the serialization mechanism in the underlying TenSEAL library.
Prime sizes that are not byte-aligned introduce additional padding overhead during (de)serialization.
To validate this explanation, we benchmark CKKS serialization and encryption runtime while varying the outer prime size from 20 to 40 bits.
Results in Table~\ref{tab:ckks_enc_time} support this claim, showing that byte-aligned prime sizes are generally faster than non-aligned ones, even when using larger bit lengths.



\begin{Takeaway}{Impact on Computation}
xIn \acro, the choice of CKKS inner and outer prime sizes significantly impacts computation overhead, with larger values leading to higher runtime. 
However, this effect is only evident when CKKS encryption/decryption operations dominate the computation; otherwise, local training masks this trend. 
\end{Takeaway}

\begin{figure*}[t]
\centering

\begin{subfigure}[b]{0.3\textwidth}
    \centering
    \includegraphics[width=\textwidth]{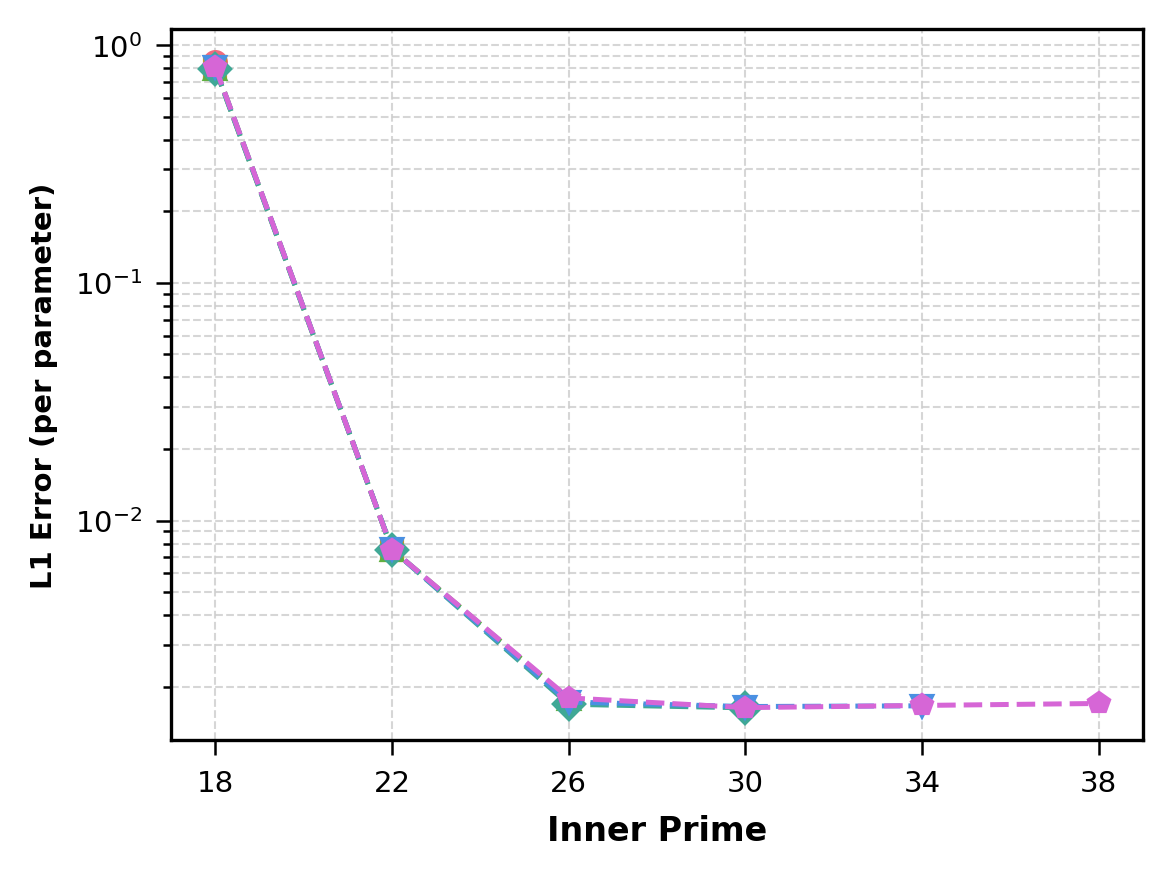}
    \caption{\finetune\ -- FEMNIST}
    \label{fig:ft_l1_femnist}
\end{subfigure}
\hfill
\begin{subfigure}[b]{0.3\textwidth}
    \centering
    \includegraphics[width=\textwidth]{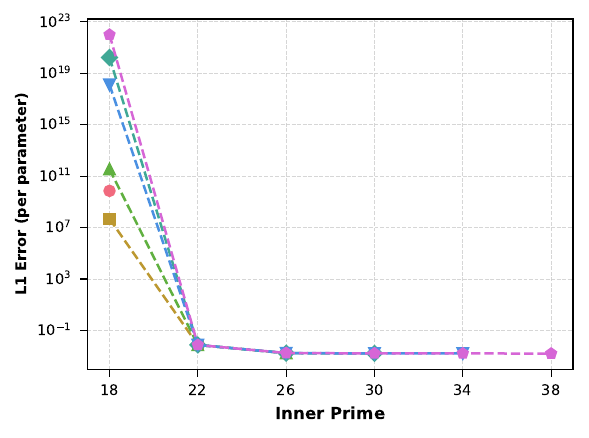}
    \caption{\ditto\ -- FEMNIST}
    \label{fig:ditto_l1_femnist}
\end{subfigure}
\hfill
\begin{subfigure}[b]{0.3\textwidth}
    \centering
    \includegraphics[width=\textwidth]{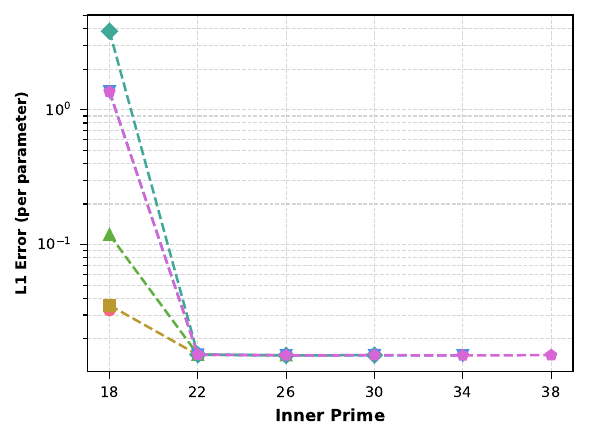}
    \caption{\fedper\ -- FEMNIST}
    \label{fig:fedper_l1_femnist}
\end{subfigure}

\vspace{0.5cm}
\begin{subfigure}[b]{0.3\textwidth}
    \centering
    \includegraphics[width=\textwidth]{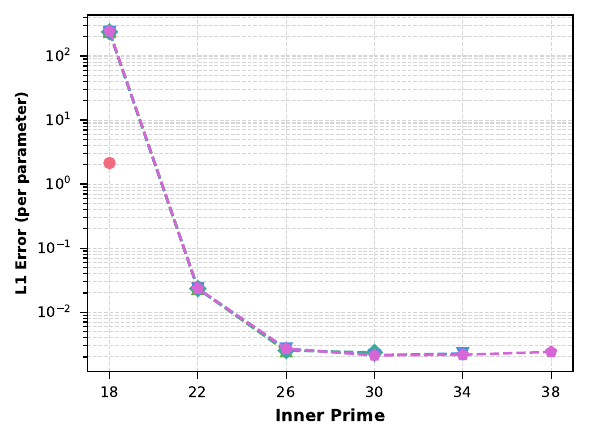}
    \caption{\finetune\ -- CelebA}
    \label{fig:ft_l1_celebA}
\end{subfigure}
\hfill
\begin{subfigure}[b]{0.3\textwidth}
    \centering
    \includegraphics[width=\textwidth]{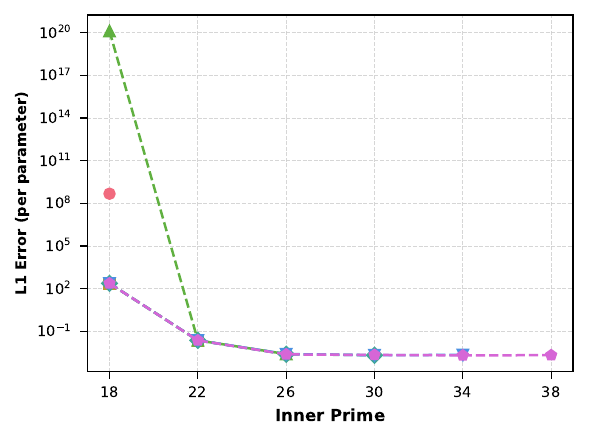}
    \caption{\ditto\ -- CelebA}
    \label{fig:ditto_l1_celebA}
\end{subfigure}
\hfill
\begin{subfigure}[b]{0.3\textwidth}
    \centering
    \includegraphics[width=\textwidth]{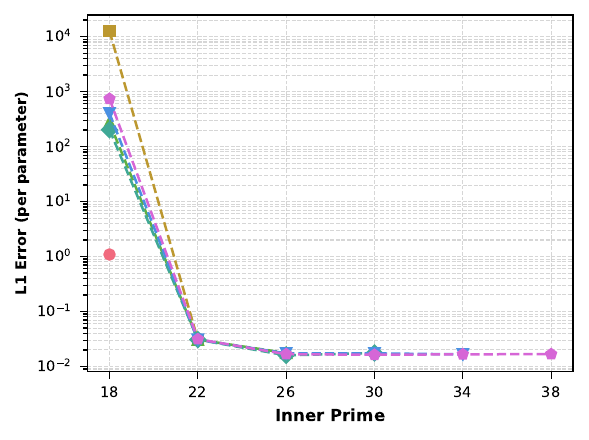}
    \caption{\fedper\ -- CelebA}
    \label{fig:fedper_l1_celebA}
\end{subfigure}

\vspace{0.5cm}
\begin{subfigure}[b]{0.3\textwidth}
    \centering
    \includegraphics[width=\textwidth]{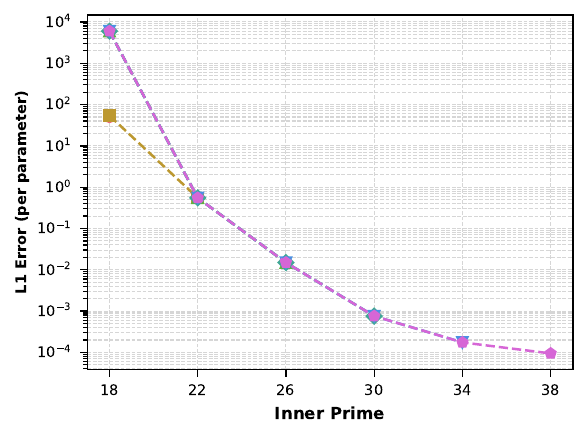}
    \caption{\finetune\ -- Sentiment140}
    \label{fig:ft_l1_sentiment}
\end{subfigure}
\hfill
\begin{subfigure}[b]{0.3\textwidth}
    \centering
    \includegraphics[width=\textwidth]{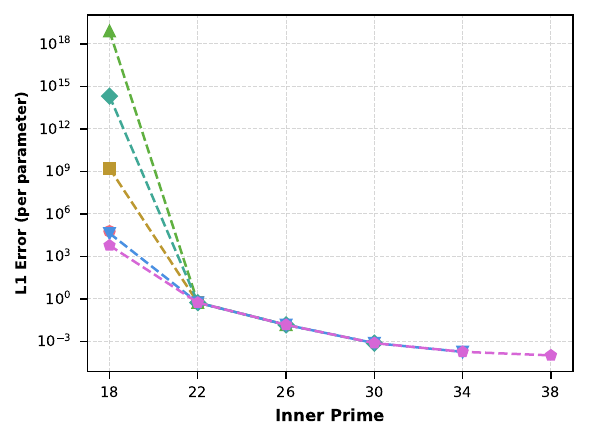}
    \caption{\ditto\ -- Sentiment140}
    \label{fig:ditto_l1_sentiment}
\end{subfigure}
\hfill
\begin{subfigure}[b]{0.3\textwidth}
    \centering
    \includegraphics[width=\textwidth]{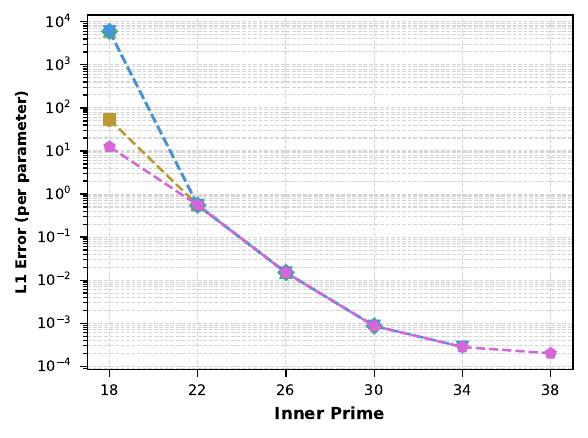}
    \caption{\fedper\ -- Sentiment140}
    \label{fig:fedper_l1sentiment}
\end{subfigure}

\caption{L1 error (accuracy deviation) of \acro\ under various inner and outer prime bits across FEMNIST, CelebA, and Sentiment datasets. Each line corresponds to a different \emph{outer} prime value:
\textcolor{outer20}{\textbf{$\bullet$=20}},
\textcolor{outer24}{\textbf{$\blacksquare$=24}},
\textcolor{outer28}{\textbf{$\blacktriangle$=28}},
\textcolor{outer32}{\textbf{$\blacklozenge$=32}},
\textcolor{outer36}{\textbf{$\blacktriangledown$=36}},
\textcolor{outer40}{\textbf{\blackpentagon=40}};  }
\label{fig:accuracy_trend}
\end{figure*}

\definecolor{cellred}{RGB}{255, 102, 102}       
\definecolor{cellorange}{RGB}{255, 178, 102}    
\definecolor{celllightgreen}{RGB}{198, 239, 206}
\definecolor{celldarkgreen}{RGB}{0, 176, 80}    

\newcommand{\red}[1]{\cellcolor{cellred}$#1$}
\newcommand{\ora}[1]{\cellcolor{cellorange}$#1$}
\newcommand{\lgrn}[1]{\cellcolor{celllightgreen}$#1$}
\newcommand{\dgrn}[1]{\cellcolor{celldarkgreen}$#1$}
\begin{table*}[ht]
\centering
\begin{adjustbox}{width=0.33\textwidth,valign=t}
\begin{tabular}{c|cccccc}
\toprule
\multicolumn{7}{c}{\textbf{FEMNIST -- Fine-Tuning (Base: 83.4\%)}} \\
\midrule
\textbf{O$\backslash$I} & \textbf{18} & \textbf{22} & \textbf{26} & \textbf{30} & \textbf{34} & \textbf{38} \\
\midrule
\textbf{20} & \red{-81.4{\pm}1.2} & $-$ & $-$ & $-$ & $-$ & $-$ \\
\textbf{24} & \red{-79.3{\pm}2.0} & \lgrn{-0.9{\pm}0.2} & $-$ & $-$ & $-$ & $-$ \\
\textbf{28} & \red{-78.9{\pm}2.3} & \lgrn{-0.8{\pm}0.3} & \lgrn{-0.4{\pm}0.3} & $-$ & $-$ & $-$ \\
\textbf{32} & \red{-76.5{\pm}0.9} & \lgrn{-0.8{\pm}0.2} & \lgrn{-0.1{\pm}0.1} & \lgrn{-0.1{\pm}0.1} & $-$ & $-$ \\
\textbf{36} & \red{-78.3{\pm}1.0} & \red{-1.2{\pm}0.3} & \lgrn{-0.2{\pm}0.1} & \lgrn{-0.1{\pm}0.2} & \lgrn{-0.2{\pm}0.3} & $-$ \\
\textbf{40} & \red{-77.7{\pm}0.6} & \lgrn{-0.8{\pm}0.4} & \lgrn{-0.3{\pm}0.1} & \lgrn{-0.5{\pm}0.2} & \lgrn{-0.2{\pm}0.2} & \lgrn{-0.4{\pm}0.2} \\
\bottomrule
\end{tabular}
\end{adjustbox}%
\begin{adjustbox}{width=0.33\textwidth,valign=t}
\begin{tabular}{c|cccccc}
\toprule
\multicolumn{7}{c}{\textbf{FMNIST -- Ditto (Base: 83.1\%)}} \\
\midrule
\textbf{O$\backslash$I} & \textbf{18} & \textbf{22} & \textbf{26} & \textbf{30} & \textbf{34} & \textbf{38} \\
\midrule
\textbf{20} & \red{-78.3{\pm}0.9} & $-$ & $-$ & $-$ & $-$ & $-$ \\
\textbf{24} & \red{-78.9{\pm}0.7} & \lgrn{-0.4{\pm}0.2} & $-$ & $-$ & $-$ & $-$ \\
\textbf{28} & \red{-79.5{\pm}1.3} & \lgrn{-0.5{\pm}0.4} & \dgrn{+0.2{\pm}0.4} & $-$ & $-$ & $-$ \\
\textbf{32} & \red{-78.7{\pm}1.1} & \lgrn{-0.6{\pm}0.1} & \dgrn{+0.1{\pm}0.6} & \dgrn{+0.5{\pm}0.0} & $-$ & $-$ \\
\textbf{36} & \red{-78.9{\pm}0.6} & \lgrn{-0.4{\pm}0.2} & \dgrn{+0.2{\pm}0.3} & \dgrn{+0.3{\pm}0.2} & \lgrn{-0.1{\pm}0.2} & $-$ \\
\textbf{40} & \red{-78.4{\pm}0.2} & \lgrn{-0.4{\pm}0.2} & \lgrn{-0.0{\pm}0.2} & \lgrn{-0.1{\pm}0.1} & \lgrn{-0.4{\pm}0.2} & \dgrn{+0.0{\pm}0.1} \\
\bottomrule
\end{tabular}
\end{adjustbox}%
\begin{adjustbox}{width=0.33\textwidth,valign=t}
\begin{tabular}{c|cccccc}
\toprule
\multicolumn{7}{c}{\textbf{FMNIST -- FedPer (Base: 81.7\%)}} \\
\midrule
\textbf{O$\backslash$I} & \textbf{18} & \textbf{22} & \textbf{26} & \textbf{30} & \textbf{34} & \textbf{38} \\
\midrule
\textbf{20} & \red{-76.7{\pm}0.2} & $-$ & $-$ & $-$ & $-$ & $-$ \\
\textbf{24} & \red{-75.9{\pm}1.2} & \lgrn{-0.1{\pm}0.1} & $-$ & $-$ & $-$ & $-$ \\
\textbf{28} & \red{-75.7{\pm}0.9} & \lgrn{-0.5{\pm}0.2} & \dgrn{+0.0{\pm}0.4} & $-$ & $-$ & $-$ \\
\textbf{32} & \red{-76.3{\pm}0.8} & \lgrn{-0.2{\pm}0.5} & \lgrn{-0.3{\pm}0.5} & \dgrn{+0.1{\pm}0.1} & $-$ & $-$ \\
\textbf{36} & \red{-76.2{\pm}0.2} & \lgrn{-0.5{\pm}0.5} & \lgrn{-0.2{\pm}0.2} & \lgrn{-0.1{\pm}0.3} & \dgrn{+0.1{\pm}0.5} & $-$ \\
\textbf{40} & \red{-74.8{\pm}1.4} & \red{-1.4{\pm}0.4} & \dgrn{+0.3{\pm}0.5} & \dgrn{+0.1{\pm}0.4} & \dgrn{+0.2{\pm}0.5} & \dgrn{+0.4{\pm}0.4} \\
\bottomrule
\end{tabular}
\end{adjustbox}

\vspace{8pt}

\begin{adjustbox}{width=0.33\textwidth,valign=t}
\begin{tabular}{c|cccccc}
\toprule
\multicolumn{7}{c}{\textbf{Sentiment140 -- Fine-Tuning (Base: 67.1\%)}} \\
\midrule
\textbf{O$\backslash$I} & \textbf{18} & \textbf{22} & \textbf{26} & \textbf{30} & \textbf{34} & \textbf{38} \\
\midrule
\textbf{20} & \red{-17.4{\pm}0.4} & $-$ & $-$ & $-$ & $-$ & $-$ \\
\textbf{24} & \red{-17.1{\pm}0.1} & \red{-2.1{\pm}0.4} & $-$ & $-$ & $-$ & $-$ \\
\textbf{28} & \red{-17.1{\pm}0.0} & \red{-1.7{\pm}0.3} & \dgrn{+0.0{\pm}0.3} & $-$ & $-$ & $-$ \\
\textbf{32} & \red{-17.0{\pm}0.0} & \red{-2.0{\pm}0.2} & \lgrn{-0.2{\pm}0.3} & \dgrn{+0.1{\pm}0.1} & $-$ & $-$ \\
\textbf{36} & \red{-16.6{\pm}0.6} & \red{-1.7{\pm}0.4} & \dgrn{+0.4{\pm}0.1} & \dgrn{+0.0{\pm}0.1} & \dgrn{+0.3{\pm}0.3} & $-$ \\
\textbf{40} & \red{-17.0{\pm}0.0} & \red{-2.0{\pm}0.3} & \lgrn{-0.0{\pm}0.4} & \dgrn{+0.3{\pm}0.2} & \dgrn{+0.1{\pm}0.5} & \dgrn{+0.2{\pm}0.4} \\
\bottomrule
\end{tabular}
\end{adjustbox}%
\begin{adjustbox}{width=0.33\textwidth,valign=t}
\begin{tabular}{c|cccccc}
\toprule
\multicolumn{7}{c}{\textbf{Sentiment140 -- Ditto (Base: 68.2\%)}} \\
\midrule
\textbf{O$\backslash$I} & \textbf{18} & \textbf{22} & \textbf{26} & \textbf{30} & \textbf{34} & \textbf{38} \\
\midrule
\textbf{20} & \red{-18.3{\pm}0.4} & $-$ & $-$ & $-$ & $-$ & $-$ \\
\textbf{24} & \red{-18.0{\pm}0.4} & \red{-2.6{\pm}0.6} & $-$ & $-$ & $-$ & $-$ \\
\textbf{28} & \red{-17.8{\pm}0.0} & \red{-1.8{\pm}0.7} & \lgrn{-0.2{\pm}0.5} & $-$ & $-$ & $-$ \\
\textbf{32} & \red{-18.0{\pm}0.4} & \red{-2.2{\pm}0.5} & \lgrn{-0.3{\pm}0.5} & \lgrn{-0.2{\pm}0.3} & $-$ & $-$ \\
\textbf{36} & \red{-17.8{\pm}0.0} & \red{-2.3{\pm}0.6} & \lgrn{-0.4{\pm}0.4} & \lgrn{-0.1{\pm}0.7} & \lgrn{-0.4{\pm}0.5} & $-$ \\
\textbf{40} & \red{-17.8{\pm}0.0} & \red{-2.3{\pm}0.6} & \lgrn{-0.7{\pm}0.6} & \lgrn{-0.3{\pm}0.2} & \lgrn{-0.1{\pm}0.7} & \lgrn{-0.2{\pm}0.4} \\
\bottomrule
\end{tabular}
\end{adjustbox}%
\begin{adjustbox}{width=0.33\textwidth,valign=t}
\begin{tabular}{c|cccccc}
\toprule
\multicolumn{7}{c}{\textbf{Sentiment140 -- FedPer (Base: 66.4\%)}} \\
\midrule
\textbf{O$\backslash$I} & \textbf{18} & \textbf{22} & \textbf{26} & \textbf{30} & \textbf{34} & \textbf{38} \\
\midrule
\textbf{20} & \red{-16.3{\pm}0.0} & $-$ & $-$ & $-$ & $-$ & $-$ \\
\textbf{24} & \red{-16.3{\pm}0.0} & \red{-1.9{\pm}0.2} & $-$ & $-$ & $-$ & $-$ \\
\textbf{28} & \red{-16.4{\pm}0.1} & \red{-1.8{\pm}0.7} & \dgrn{+0.0{\pm}0.1} & $-$ & $-$ & $-$ \\
\textbf{32} & \red{-16.3{\pm}0.0} & \red{-1.4{\pm}0.2} & \dgrn{+0.4{\pm}0.5} & \lgrn{-0.2{\pm}0.2} & $-$ & $-$ \\
\textbf{36} & \red{-16.3{\pm}0.0} & \red{-1.9{\pm}0.2} & \lgrn{-0.5{\pm}0.3} & \dgrn{+0.1{\pm}0.4} & \lgrn{-0.2{\pm}0.3} & $-$ \\
\textbf{40} & \red{-16.3{\pm}0.0} & \red{-1.7{\pm}0.5} & \dgrn{+0.3{\pm}0.2} & \dgrn{+0.5{\pm}0.2} & \lgrn{-0.5{\pm}0.3} & \lgrn{-0.3{\pm}0.1} \\
\bottomrule
\end{tabular}
\end{adjustbox}

\vspace{8pt}

\begin{adjustbox}{width=0.33\textwidth,valign=t}
\begin{tabular}{c|cccccc}
\toprule
\multicolumn{7}{c}{\textbf{CelebA -- Fine-Tuning (Base: 86.1\%)}} \\
\midrule
\textbf{O$\backslash$I} & \textbf{18} & \textbf{22} & \textbf{26} & \textbf{30} & \textbf{34} & \textbf{38} \\
\midrule
\textbf{20} & \red{-36.6{\pm}2.2} & $-$ & $-$ & $-$ & $-$ & $-$ \\
\textbf{24} & \red{-32.8{\pm}4.6} & \dgrn{+0.4{\pm}2.6} & $-$ & $-$ & $-$ & $-$ \\
\textbf{28} & \red{-35.3{\pm}2.4} & \dgrn{+1.1{\pm}1.1} & \dgrn{+0.1{\pm}1.6} & $-$ & $-$ & $-$ \\
\textbf{32} & \red{-35.2{\pm}2.5} & \dgrn{+0.7{\pm}1.7} & \dgrn{+0.2{\pm}1.6} & \lgrn{-0.3{\pm}1.5} & $-$ & $-$ \\
\textbf{36} & \red{-35.3{\pm}2.4} & \dgrn{+1.2{\pm}1.2} & \dgrn{+0.3{\pm}1.8} & \lgrn{-0.2{\pm}1.7} & \dgrn{+0.2{\pm}1.7} & $-$ \\
\textbf{40} & \red{-33.5{\pm}0.2} & \dgrn{+1.4{\pm}1.2} & \dgrn{+0.2{\pm}1.7} & \dgrn{+0.1{\pm}2.1} & \lgrn{-0.1{\pm}2.1} & \dgrn{+0.0{\pm}1.5} \\
\bottomrule
\end{tabular}
\end{adjustbox}%
\begin{adjustbox}{width=0.33\textwidth,valign=t}
\begin{tabular}{c|cccccc}
\toprule
\multicolumn{7}{c}{\textbf{CelebA -- Ditto (Base: 86.3\%)}} \\
\midrule
\textbf{O$\backslash$I} & \textbf{18} & \textbf{22} & \textbf{26} & \textbf{30} & \textbf{34} & \textbf{38} \\
\midrule
\textbf{20} & \red{-37.7{\pm}1.8} & $-$ & $-$ & $-$ & $-$ & $-$ \\
\textbf{24} & \red{-36.7{\pm}2.0} & \dgrn{+0.3{\pm}0.5} & $-$ & $-$ & $-$ & $-$ \\
\textbf{28} & \red{-37.4{\pm}2.3} & \dgrn{+0.9{\pm}0.6} & \lgrn{-0.2{\pm}1.7} & $-$ & $-$ & $-$ \\
\textbf{32} & \red{-36.3{\pm}2.3} & \dgrn{+1.0{\pm}1.5} & \lgrn{-0.1{\pm}1.7} & \lgrn{-0.1{\pm}1.8} & $-$ & $-$ \\
\textbf{36} & \red{-33.8{\pm}0.2} & \dgrn{+1.1{\pm}1.4} & \dgrn{+0.1{\pm}1.5} & \lgrn{-0.1{\pm}1.8} & \lgrn{-0.1{\pm}1.8} & $-$ \\
\textbf{40} & \red{-35.6{\pm}2.5} & \dgrn{+0.7{\pm}1.2} & \lgrn{-0.4{\pm}1.6} & \lgrn{-0.4{\pm}1.5} & \lgrn{-0.3{\pm}2.2} & \lgrn{-0.0{\pm}2.4} \\
\bottomrule
\end{tabular}
\end{adjustbox}%
\begin{adjustbox}{width=0.33\textwidth,valign=t}
\begin{tabular}{c|cccccc}
\toprule
\multicolumn{7}{c}{\textbf{CelebA -- FedPer (Base: 86.5\%)}} \\
\midrule
\textbf{O$\backslash$I} & \textbf{18} & \textbf{22} & \textbf{26} & \textbf{30} & \textbf{34} & \textbf{38} \\
\midrule
\textbf{20} & \red{-34.0{\pm}0.2} & $-$ & $-$ & $-$ & $-$ & $-$ \\
\textbf{24} & \red{-34.0{\pm}0.2} & \lgrn{-0.3{\pm}0.8} & $-$ & $-$ & $-$ & $-$ \\
\textbf{28} & \red{-35.8{\pm}2.4} & \lgrn{-0.4{\pm}1.1} & \lgrn{-0.3{\pm}1.5} & $-$ & $-$ & $-$ \\
\textbf{32} & \red{-33.9{\pm}0.2} & \dgrn{+0.4{\pm}0.5} & \lgrn{-0.8{\pm}1.5} & \dgrn{+0.3{\pm}0.6} & $-$ & $-$ \\
\textbf{36} & \red{-34.0{\pm}0.2} & \lgrn{-0.1{\pm}1.4} & \lgrn{-0.9{\pm}0.2} & \lgrn{-0.0{\pm}1.1} & \dgrn{+0.7{\pm}1.0} & $-$ \\
\textbf{40} & \red{-34.0{\pm}0.2} & \lgrn{-0.9{\pm}2.4} & \dgrn{+0.1{\pm}1.5} & \lgrn{-0.4{\pm}0.3} & \lgrn{-0.1{\pm}0.3} & \lgrn{-0.1{\pm}0.9} \\
\bottomrule
\end{tabular}
\end{adjustbox}

\vspace{4pt}
\noindent
\colorbox{cellred}{\strut~} $< -1\%$ 
\colorbox{celllightgreen}{\strut~} $-1\%$ to $0\%$ \quad
\colorbox{celldarkgreen}{\color{white}\strut~} $> 0\%$

\caption{Accuracy difference (\acro vs the baseline) in percentage points for each dataset and federated learning method. O and I denote the outer and inner prime parameters, respectively. Values are reported as mean $\pm$ std across 3 seeds.}
\label{tab:all_accuracy_diff}
\end{table*}

\subsection{Precision}\label{subsec:precision-overhead}
We evaluate the impact of \acro on precision at two levels: the parameter level and the model level.
This allows us to assess how \acro overhead affects the numerical precision of the model parameters (the former) and whether these parameter deviations are sufficient to impact the model's predictive performance (the latter).

\paragraph{Parameter-level Precision}
We define the parameter-level precision overhead as the average $\ell_1$ training error ($\mathcal{E}_{\ell1}$) between the final personalized parameters in \acro ($\mathbf{w}^{\text{\acro}}$) and the baseline ($\mathbf{w}^{\text{baseline}}$), formally:
\begin{equation}
    \mathcal{E}_{\ell1} = \frac{1}{N} \sum_{i=1}^{N} \left\| \mathbf{w}^{\text{\acro}}_i - \mathbf{w}^{\text{base}}_i \right\|_1
\end{equation}

Recall from Section~\ref{subsec:analysis} that, in \acro, the precision for the decimal points of decrypted parameters is determined by the inner prime while the integer-part precision is controlled by the difference between the outer and inner prime values, where the larger value leads to better precision. 

Figure~\ref{fig:accuracy_trend} illustrates $\mathcal{E}_{\ell1}$ under various inner and outer prime bit settings. 
We observe that $\mathcal{E}_{\ell1}$ is primarily influenced by the inner prime, while the choice of outer prime has a negligible effect (besides when the inner prime is 18 bits); this trend is consistent across all nine settings.
Further inspection suggests that this behavior arises because $\mathbf{w}^{\text{base}}$ lie within a range $[-6, 6]$, which is sufficiently small for all evaluated configurations to provide enough integer precision.

In contrast, the precision of the decimal points is controlled by the inner prime. 
Increasing the inner prime improves this precision, with $\mathcal{E}_{\ell1}$ going down rapidly. 
Across all settings, this improvement begins to plateau around 26-bit inner prime (corresponding to $\mathcal{E}_{\ell1}$ around $10^{-2}$), indicating diminishing returns of $\mathcal{E}_{\ell1}$ afterwards.

\paragraph{Model-level Precision}
Next, we consider model-level precision by measuring the difference in test accuracy between \acro and the baseline. Table~\ref{tab:all_accuracy_diff} reports these results.
We observe that, regardless of the outer prime value, using an 18-bit inner prime leads to severe model accuracy degradation. 
This behavior is consistent with the parameter-level precision results in Figure~\ref{fig:accuracy_trend}, which show that the personalized parameters obtained with an 18-bit inner prime deviate substantially from the baseline. 
In some cases (e.g., Figure~\ref{fig:ditto_l1_femnist}), the parameter error $\mathcal{E}_{\ell1}$ reaches as high as $10^{22}$. Such extreme errors suggest that the accumulated CKKS approximation noise may cause overflow during decoding, resulting in corrupted model parameters and poor model accuracy.

In contrast, increasing the inner prime to 22 bits substantially improves compared to the 18-bit setting. However, the test accuracy still shows an observable drop (more than 1\%) in many settings. 
In comparison, larger inner primes ($\ge$ 26 bits) consistently maintain test accuracy at a level comparable to the baseline.
Similar to parameter-level precision, the outer prime has no impact on model-level precision.

\begin{figure*}[t]
\centering

\begin{subfigure}[b]{0.3\textwidth}
    \centering
    \includegraphics[width=\textwidth]{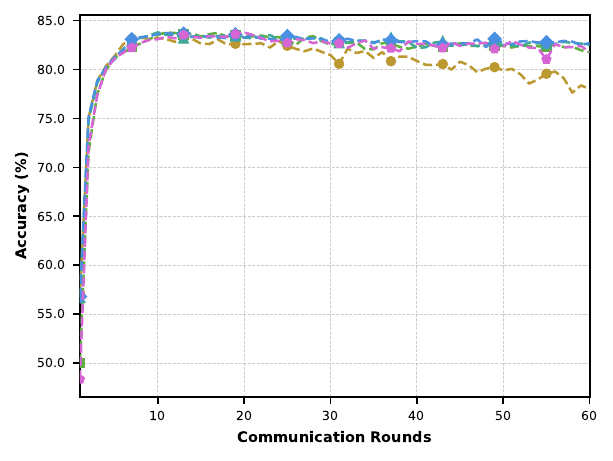}
    \caption{\finetune\ -- FEMNIST}
    \label{fig:ft_com_round_femnist}
\end{subfigure}
\hfill
\begin{subfigure}[b]{0.3\textwidth}
    \centering
    \includegraphics[width=\textwidth]{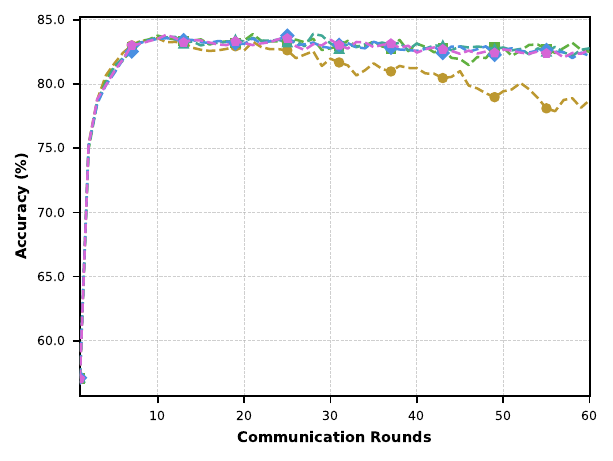}
    \caption{\ditto\ -- FEMNIST}
    \label{fig:ditto_com_round_femnist}
\end{subfigure}
\hfill
\begin{subfigure}[b]{0.3\textwidth}
    \centering
    \includegraphics[width=\textwidth]{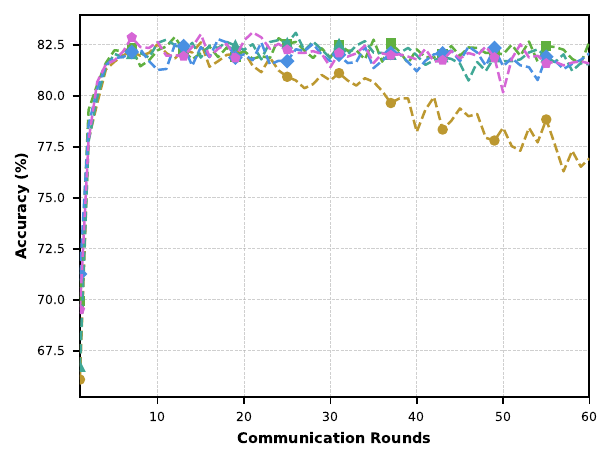}
    \caption{\fedper\ -- FEMNIST}
    \label{fig:fedper_com_round_femnist}
\end{subfigure}

\vspace{0.5cm}
\begin{subfigure}[b]{0.3\textwidth}
    \centering
    \includegraphics[width=\textwidth]{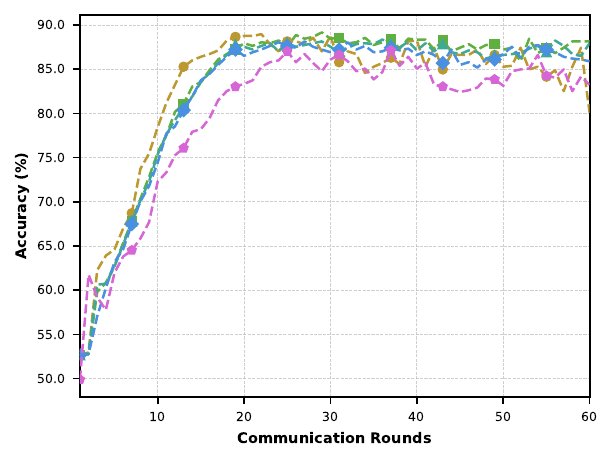}
    \caption{\finetune\ -- CelebA}
    \label{fig:ft_com_round_celebA}
\end{subfigure}
\hfill
\begin{subfigure}[b]{0.3\textwidth}
    \centering
    \includegraphics[width=\textwidth]{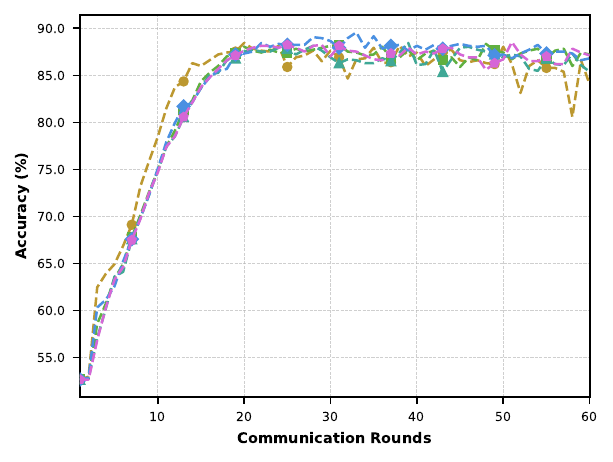}
    \caption{\ditto\ -- CelebA}
    \label{fig:ditto_com_round_celebA}
\end{subfigure}
\hfill
\begin{subfigure}[b]{0.3\textwidth}
    \centering
    \includegraphics[width=\textwidth]{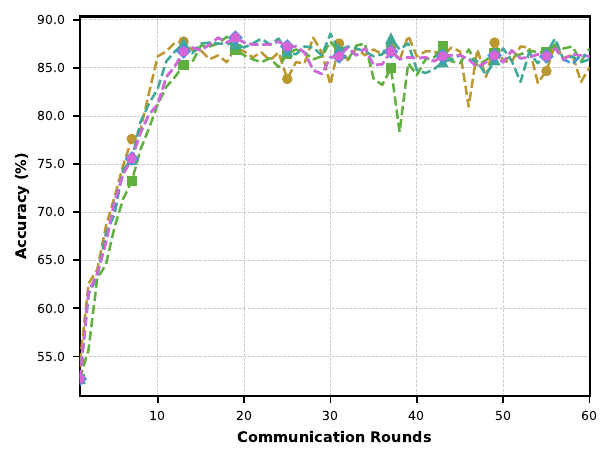}
    \caption{\fedper\ -- CelebA}
    \label{fig:fedper_com_round_celebA}
\end{subfigure}

\vspace{0.5cm}
\begin{subfigure}[b]{0.3\textwidth}
    \centering
    \includegraphics[width=\textwidth]{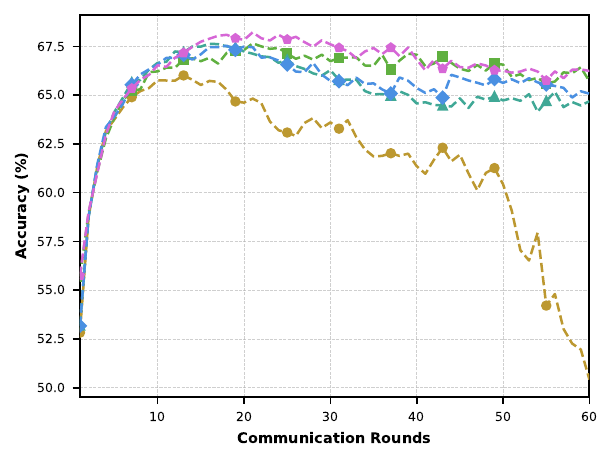}
    \caption{\finetune\ -- Sentiment140}
    \label{fig:ft_com_round_sentiment}
\end{subfigure}
\hfill
\begin{subfigure}[b]{0.3\textwidth}
    \centering
    \includegraphics[width=\textwidth]{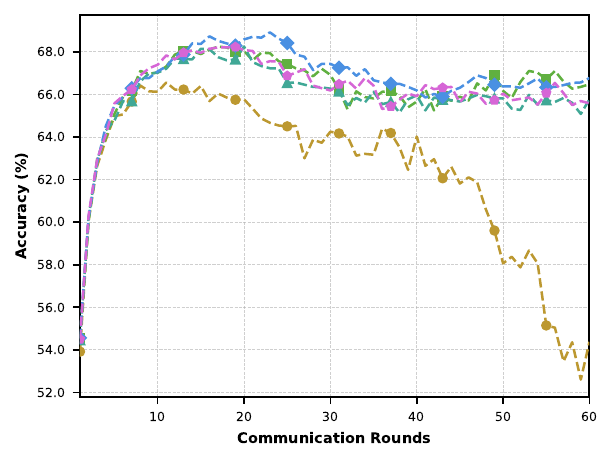}
    \caption{\ditto\ -- Sentiment140}
    \label{fig:ditto_com_round_sentiment}
\end{subfigure}
\hfill
\begin{subfigure}[b]{0.3\textwidth}
    \centering
    \includegraphics[width=\textwidth]{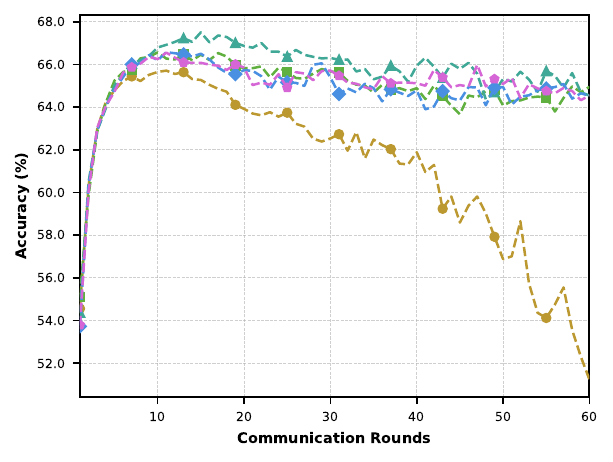}
    \caption{\fedper\ -- Sentiment140}
    \label{fig:fedper_com_round_sentiment}
\end{subfigure}

\caption{Model accuracy across different numbers of communication rounds when \acro is instantiated with a fixed 40-bit outer prime and varying inner prime bits. Each line represents a different inner prime value:
\textcolor{inner22}{\textbf{$\bullet$=22}},
\textcolor{inner26}{\textbf{$\blacksquare$=26}},
\textcolor{inner30}{\textbf{$\blacktriangle$=30}},
\textcolor{inner34}{\textbf{$\blacklozenge$=34}},
\textcolor{inner38}{\textbf{\blackpentagon=38}}.
}
\label{fig:commrounds}
\end{figure*}


Next, we investigate the impact of the number of federated communication rounds on both model convergence and model-level precision in \acro. 
Since CKKS introduces approximation noise, repeated encryption and decryption across rounds may accumulate errors and potentially affect convergence/accuracy.

Since the outer prime has negligible impact on precision, we fix the outer prime to 40 bits and vary only the inner prime; we also intentionally exclude the 18-bit inner prime from this experiment due to its poor accuracy results in the previous evaluation. 
We then track model accuracy over 60 federated communication rounds.

Figure~\ref{fig:commrounds} presents the results. 
Overall, almost all datasets and PFL algorithms exhibit similar convergence behavior: model accuracy improves rapidly during the early communication rounds and gradually stabilizes as training progresses.
The only notable exceptions are FEMNIST and Sentiment140 when using a 22-bit inner prime. 
In these cases, model accuracy begins to decline after 25 communication rounds for FEMNIST and 10 communication rounds for Sentiment140. 
A possible explanation is the accumulation of CKKS noise over repeated rounds of training, which may gradually reduce the precision of the aggregated model and negatively affect convergence.
Since both models contain a relatively large number of parameters, they may be more susceptible to the accumulation of CKKS approximation noise across communication rounds, resulting in degraded accuracy during the later stages of training.
This result further supports our earlier finding that a 22-bit inner prime is insufficient to maintain model-level precision. 
In contrast, inner primes of at least 26 bits consistently achieve convergence behavior and final model accuracy comparable to the baseline, even after prolonged training.

\begin{Takeaway}{Impact on Precision}
xIn \acro, the inner prime is the primary factor affecting both parameter- and model-level precision. Smaller inner primes increase average $\ell_1$ parameter error, introduce numerical instability, and degrade test accuracy. Increasing the inner prime improves precision significantly, while diminishing returns are observed beyond 26 bits. 
In contrast, the outer prime has negligible impact on precision.
\end{Takeaway}


\subsection{Parameters Recommendation}

Our empirical results highlight the trade-off between precision and computational/communication costs in \acro. In particular, Takeaways~1-2 suggest using smaller inner and outer primes to reduce computation and communication overhead.

In contrast, Takeaway 3 indicates that increasing the inner prime improves parameter/model-level precision while the outer prime has negligible impact on precision. 
In practice, however, minor degradation at the parameter level is acceptable as long as model-level accuracy is preserved. 
Our results show that inner primes of 26 bits or higher consistently maintain model-level accuracy comparable to the baseline across all evaluated outer prime settings. 
While larger inner primes further reduce parameter-level error, they yield negligible improvements in model-level accuracy beyond 26 bits, while incurring higher computation and communication costs.

Based on these observations, we recommend using a CKKS ciphertext modulus configuration of $(28,26,28)$ for \acro deployment. This configuration minimizes communication and computation overhead while maintaining model accuracy comparable to the baseline.

\section{Related Work}
\label{sec:related}
Despite not sharing raw data directly, traditional FL methods still suffer from the risk of privacy leakage through exchanged model updates~\cite{zhu2019deep,wang2019beyond,melis2019exploiting,luo2021feature,nasr2019comprehensive,shokri2017membership}.
To mitigate this impact, prior research has explored integrating privacy-enhancing technologies (PETs) into FL.
These efforts generally fall into four categories based on the underlying PET technique:
(1) Homomorphic Encryption (HE)~\cite{wang2020achieve,liu2021privacy},
(2) Secure Multi-party Computation~\cite{segal2017practical,ma2023flamingo},
(3) Differential Privacy~\cite{shi2023make,truex2019hybrid}, and
(4) Zero-Knowledge Proofs~\cite{li2021privacy}.
Each category offers distinct trade-offs in terms of computation and communication costs, trust assumptions, and security guarantees; we refer to the recent survey~\cite{chen2024federated} for a comprehensive comparison of these categories.

Since this work focuses on integrating the CKKS scheme into PFL, we review prior work in the HE category.
Early work~\cite{cheng2021secureboost,aono2017privacy,liu2019secure} adopted the Paillier additive HE scheme~\cite{paillier1999public} to encrypt local model updates before aggregation on the server.
However, these approaches introduced substantial computation and communication overhead, which impeded their practical applicability.
Subsequent studies~\cite{zhang2020batchcrypt,li2025fedphe,pan2024fedshe,qiu2022privacy} improved upon this by incorporating batching techniques to group model parameters prior to encryption.
One promising direction is to adopt the CKKS scheme~\cite{cheon2017homomorphic,cheon2018full} into FL since it natively supports batching by encoding multiple floating-point parameters into a single plaintext, later encrypted into a single ciphertext.

Building upon this property, FedSHE~\cite{pan2024fedshe} investigates the integration of CKKS into FL, analyzing how the multiplicative depth parameter ($L$) affects computational overhead.
However, its analysis is limited to a single parameter while leaving other CKKS parameters unexplored.
In contrast, our parameter selection analysis is more comprehensive, covering a wider range of CKKS parameters.
Moreover, our evaluation considers both computation and communication costs, and more importantly their trade-off with decryption precision, which was not explored in FedSHE.
Finally, FedSHE focuses on standard FL with the FedAVG algorithm~\cite{mcmahan2017communication}, whereas our work targets PFL with the analysis and findings that can generalize to standard FL settings.

To the best of our knowledge, only one study has examined the integration of HE within a PFL framework:
CLDP-MPE~\cite{shen2024ckks} applies the CKKS scheme and local differential privacy to enhance security and privacy of the PFL method based on a group-based meta-learning algorithm~\cite{yang2023personalized}.
Nonetheless, it does not investigate how CKKS affects decryption precision or the associated computation and communication costs, which are the main focus of this work.

\section{Conclusion}
\label{sec:conclusion}
In this work, we explored the integration of CKKS homomorphic encryption into Personalized Federated Learning (PFL) through the \acro framework. 
Our evaluation of \acro across three underlying PFL algorithms (\finetune, \fedper and \ditto) demonstrated that the inner ciphertext prime is a major contributor to performance overheads and decryption precision where a larger value improves decryption accuracy but increases bandwidth and runtime. 
Based on these findings, we recommend configuring the CKKS prime parameters as $(28,26,28)$, i.e., 28-bit outer prime and 26-bit inner prime, to achieve the best practical trade-off, benefiting from high accuracy while minimizing computational and communication costs. 
Overall, \acro validates that privacy preservation in PFL can be achieved efficiently when CKKS parameters are carefully tuned. Future work includes exploring different datasets with various degrees of heterogeneity and integration with other privacy-enhancing techniques such as differential privacy.





\bibliographystyle{ieeetr}
\bibliography{ref}

\appendices

\section{Proof of \acro security}\label{apdx:proof}

Let $\mathcal{L}$ denote the protocol leakage visible to the server, including the number of participating clients, client participation, ciphertext sizes, model/update dimensions, public CKKS parameters, and dataset sizes $|D_i|$.

\begin{lemma}\label{lemma:security}
Assume that the CKKS encryption scheme used in \acro is IND-CPA secure under the RLWE assumption. Then, under the honest-but-curious server model, \acro preserves confidentiality of each client's uploaded model update $\hat{\Theta}_i$, except for the leakage $\mathcal{L}$.
\end{lemma}

\begin{proof}

We prove Lemma~\ref{lemma:security} by contradiction.
Suppose there exists a PPT adversary $\mathcal{A}$ that breaks the confidentiality guarantee stated in Lemma~\ref{lemma:security}, i.e., $\mathcal{A}$ can distinguish the encrypted upload of one client update from the encrypted upload of another equal-length client update, beyond $\mathcal{L}$, with non-negligible probability.
We show how to use $\mathcal{A}$ to construct another adversary $\mathcal{B}$ that breaks the IND-CPA security of the underlying CKKS encryption scheme.

$\mathcal{B}$ participates in the CKKS IND-CPA game and submits two equal-length plaintext updates $\hat{\Theta}_i^0$ and $\hat{\Theta}_i^1$ with the same $\mathcal{L}$ to the challenger.
The challenger samples a bit $b \in \{0,1\}$ and returns a challenge ciphertext $c^\star = \mathsf{Enc}_{pk}(\hat{\Theta}_i^b)$.
$\mathcal{B}$ then simulates the view of the honest-but-curious server for $\mathcal{A}$, using $c^\star$ as the target uploaded ciphertext $\hat{\Theta}_i^e$.
Since $\mathcal{A}$ can determine whether $c^\star$ encrypts $\hat{\Theta}_i^0$ or $\hat{\Theta}_i^1$ with non-negligible advantage, $\mathcal{B}$ can output the same guess and hence win the CKKS IND-CPA game with non-negligible advantage.

This contradicts the IND-CPA security of CKKS, which holds under the RLWE hardness assumption.
Therefore, such an adversary $\mathcal{A}$ cannot exist, and Lemma~\ref{lemma:security} follows.

\end{proof}

We note that the work in~\cite{li2021security} identified a stronger security notion, IND-CPAD$^\mathsf{D}$, where an adversary can approximate plaintext outputs when it has access to a decryption oracle. 
However, the adversary (server) is never granted decryption access, making this notion not applicable in our setting.



\end{document}